\documentclass[preprint,12pt]{elsarticle}

\usepackage{algorithmic}
\usepackage{algorithm}
\usepackage{graphicx}
\usepackage[caption=false,font=footnotesize]{subfig}
\usepackage{amsmath}
\usepackage{amssymb}
\usepackage{amsthm}
\usepackage{tabu}
\usepackage{makecell}

 \newdefinition{remark}{Remark}
 \newtheorem{lemma}{Lemma}
 \newtheorem{theorem}{Theorem}
 \newtheorem{corollary}{Corollary}
 \newdefinition{definition}{Definition}
 \newdefinition{assumption}{Assumption}

\biboptions{sort}

\journal{Ad Hoc Networks}
\begin{document}

\begin{frontmatter}

\title{On Throughput Capacity for a Class of Buffer-Limited MANETs}

\author[XD,FUN]{Jia Liu}
\ead{jliu871219@gmail.com}

\author[XD]{Min Sheng}
\ead{mshengxd@gmail.com}

\author[XD1]{Yang Xu}
\ead{yxu@xidian.edu.cn}

\author[XD]{Jiandong Li}
\ead{jdli@ieee.org}

\author[FUN]{Xiaohong Jiang}
\ead{jiang@fun.ac.jp}

\address[XD]{State Key Laboratory of ISN, Institute of Information Science, Xidian University, Xi’an, Shaanxi, 710071, China}
\address[FUN]{School of Systems Information Science, Future University Hakodate, 116-2 Kamedanakano-cho, Hakodate, Hokkaido, 041-8655, Japan}
\address[XD1]{School of Computer Science and Technology, Xidian University, Xi’an, Shaanxi, 710071, China}

\begin{abstract}
Available throughput performance studies for mobile ad hoc networks (MANETs) suffer from two major limitations: they mainly focus on the scaling law study of throughput, while the exact throughput of such networks remains largely unknown; they usually consider the infinite buffer scenarios, which are not applicable to the practical networks with limited buffer. As a step to address these limitations, this paper develops a general framework for the exact throughput capacity study of a class of buffer-limited MANETs with the two-hop relay. We first provide analysis to reveal how the throughput capacity of such a MANET is determined by its relay-buffer blocking probability (RBP). Based on the Embedded Markov Chain Theory and Queuing Theory, a novel theoretical framework is then developed to enable the RBP and closed-form expression for exact throughput capacity to be derived. We further conduct case studies under two typical transmission scheduling schemes to illustrate the applicability of our framework and to explore the corresponding capacity optimization as well as capacity scaling law. Finally, extensive simulation and numerical results are provided to validate the efficiency of our framework and to show the impacts brought by the buffer constraint.

\end{abstract}


\end{frontmatter}

\section{Introduction} \label{section:introduction}

The mobile ad hoc network (MANET) represents a kind of self-organizing network architecture, which consists of mobile nodes communicating with each other without centralized infrastructure and management \cite{Perkins_BOOK01}. Since MANET can be deployed and reconfigured rapidly at very low cost, it serves as an appealing candidate for many critical applications, such as disaster relief, battlefield communication and emergency rescue \cite{Ramanathan_CM02}. To support the design and applications of MANETs, the studies on their fundamental performance have been extensively reported. Despite much research activity, however, the lack of a general capacity theory for MANETs is still a long standing open problem and becomes an obstacle on the development and commercialization of such networks \cite{Andrews_CM08,Goldsmith_CM11}.

Since the pioneer work of Gupta and Kumar \cite{Gupta_IT00}, extensive scaling law (i.e. order sense) results on the capacity of ad hoc networks have been reported in literature, which mainly focus on the study of asymptotic per node throughput behavior as the number of network nodes increases. The results in \cite{Gupta_IT00} indicate that for a static ad hoc network, its per node throughput diminishes to zero as the number of network nodes tends to infinity. Later, Grossglauser and Tse demonstrated in \cite{Grossglauser_Tse_2001} that with the help of node mobility, a $\Theta(1)$\footnote{Please kindly refer to \cite{Cormen_BOOK01} for the notations of order results.} constant per node throughput is achievable in a MANET. Inspired by the seminar work of \cite{Grossglauser_Tse_2001}, lots of studies have been devoted to the analysis on the scaling laws of MANETs throughput under various mobility models \cite{Gamal_IT06,Lin_TON06,Mammen_IT07,Sharma_TON07,Ciullo_TON11}, and under various network scenarios \cite{Peng_INFOCOM11,Huang_TON12,Chen_INFOCOM12,Zhang_TON14}. For a survey on the scaling law results of MANETs throughput, readers are referred to \cite{Lu_2013} and references therein. Although scaling laws are helpful to grasp the general trend of MANET performance, they provide a little insight into its real achievable throughput. In practice, however, a thorough understanding on the real achievable throughput of MANETs is of great concern for network engineers, since it serves as an instruction guideline for network design and optimization.

By now, some initial and helpful results are available on the exact throughput capacity study of MANETs, i.e. to derive the exact expressions for the throughput capacity of such networks. For the regular Manhattan and ring networks, Mergen and Tong \cite{Mergen_IT2005} derived their throughput capacity in closed-form. Neely \emph{et al}. \cite{Neely_IT05} explored the exact capacity of cell-partitioned MANETs, and revealed a fundamental tradeoff between the throughput capacity and packet delay in such networks. Following this line, Wang \emph{et al}. \cite{Wang_TON11} extended the tradeoff results in \cite{Neely_IT05} to a multicast scenario. Gao \emph{et al}. \cite{Gao_2013} considered a more general network scenario with a general setting of transmission range and derived its capacity under the group-based scheduling. Liu \emph{et al}. further explored the exact throughput capacity of MANETs with packet redundancy \cite{Liu_TWC11} and power control \cite{Liu_INFOCOM12}. Recently, Chen \emph{et al}. studied the exact throughput capacity of MANETs with directional antennas \cite{Chen_2013} and explored the efficient approximations for the exact throughput capacity of MANETs with ALOHA protocol \cite{Chen_2013_Aloha}. 

It is notable, however, that one common limitation of all these studies is that to make their analysis tractable, they all assume the relay buffer of a node, which is used for temporarily storing packets of other nodes, has an infinite buffer size. This assumption does not hold for a practical MANET, where the buffer size of a mobile node is usually limited due to both its storage space limitation and computing capability limitation. Thus, for the practical capacity study of MANETs, the constraint on buffer size should be carefully addressed. Notice that the throughput capacity modeling with practical limited-buffer constraint still remains a technical challenge. This is mainly due to the lack of a general theoretical framework to efficiently characterize the highly complicated buffer occupancy behaviors in such networks.

As a step to address above limitations, this paper studies the exact throughput capacity for a class of MANETs, where each node is associated with a shared and limited relay buffer to temporarily store the packets of other nodes \cite{Herdtner_INFOCOM05,Le_TON12} and the flexible two-hop relay routing scheme is adopted for packet forwarding. The two-hop reply serves as an important routing protocol for practical MANETs \cite{Wei_2004}, since it is simple and can be implemented easily in a distributed fashion. Also, the two-hop reply is efficient in the sense that it has the capability of achieving the throughput capacity under many important MANET scenarios \cite{Neely_IT05, Gao_2013, Chen_2013_Aloha}. For this class of buffer-limited MANETs, as a thorough extension of our previous work \cite{Liu_14PIMRC}, this paper develops a general theoretical framework to enable the analytical study on their exact throughput capacity to be conducted.

The main contributions are summarized as follows:
\begin{itemize}

\item
For one buffer-limited MANET concerned in this paper, we first provide theoretical analysis to reveal the inherent relationship between its throughput capacity and its relay buffer blocking probability (RBP).

\item 
For the analysis of RBP, a novel theoretical framework based on the Embedded Markov Chain Theory is then
developed to capture the complicated queuing process in a relay buffer. With the help of the framework and also the Queueing Theory, the RBP under any exogenous input rate and the closed-form expression for exact throughput capacity are then derived.

\item
Case studies are further provided under two typical transmission scheduling schemes to illustrate how our theoretical framework can be applied for exact throughput capacity analysis under one given transmission scheduling scheme. The corresponding capacity optimization issue and scaling law performance are also explored.

\item
Finally, extensive simulation and numerical results are provided to demonstrate the efficiency of our theoretical framework on capturing the throughput behavior of a buffer-limited MANET and to illustrate the impacts brought by the buffer constraint.
\end{itemize}

The remainder of this paper is outlined as follows. Section~\ref{section:preliminaries} introduces the system models, routing scheme and definitions involved in this study. We present the overall framework for throughput capacity analysis in Section~\ref{section:throughput_capacity}, and then develop an Embedded Markov Chain-based framework in Section~\ref{section:Markov} for the evaluation of RBP and exact throughput capacity. Section~\ref{section:case_study} deals with the case studies, capacity optimization issue and scaling law results, and Section~\ref{section:simulation} provides the simulation and numerical results. Finally, we provide the related work in Section~\ref{section:related_work} and conclude this paper in Section~\ref{section:conclusion}.

\section{Preliminaries} \label{section:preliminaries}
This section presents the system models, buffer constraint, routing scheme and basic definitions involved in this study. Table~\ref{table:notations} summarizes the main notations.

\begin{table*}[t]
\caption{Main notations}
\renewcommand\arraystretch{1.1}
\begin{tabu} to \hsize {X X[8,l]}
\hline 
Symbol             & Quantity \tabularnewline \Xhline{1.2pt}
$n$                & Number of nodes                \tabularnewline
$m$                & The network is partitioned into $m\times m$ cells     \tabularnewline
$B$                & Relay buffer size              \tabularnewline
$\lambda$          & Exogenous packet arrival rate         \tabularnewline
$p_b(\lambda)$     & Relay-buffer blocking probability \tabularnewline
$T_c$              & Throughput capacity            \tabularnewline
$\alpha$           & The probability that a node selects to execute S-R transmission when it gets access to the wireless channel and cannot execute the S-D
                     transmission \tabularnewline 
\hline

\end{tabu}
\label{table:notations}
\end{table*}

\subsection{System Models}  \label{subsection:system_model}

As illustrated in Fig.~\ref{fig:system_models} that we consider a time-slotted and cell-partitioned network with $n$ mobile nodes and $m\times m$ non-overlapping equal cells, where the nodes roam from cell to cell over the network according to the independent and identically distributed (i.i.d) mobility model \cite{Neely_IT05,Sharma_TON07}, and during a time slot the total amount of data that can be transmitted from a node to another is fixed and normalized to one packet. With the i.i.d mobility model, each node independently selects a cell among all cells with equal probability at the beginning of each time slot and then stays in it during this time slot, so the location of each node is i.i.d and uniformly distributed over all cells in each time slot, and between time slots the distributions of nodes' locations are independent. 

We adopt the i.i.d. mobility model here mainly due to the following reasons. First, the mathematical tractability of this model allows us to gain important insights into the structure of throughput capacity analysis. Second, as illustrated in \cite{Neely_IT05} and to be demonstrated in Section~\ref{subsection:simulation_results}, the throughput capacity result derived under this model can be applied to other typical mobility models like the random walk model. Third, the analysis under this model provides a meaningful theoretical performance result in the limit of infinite mobility \cite{Neely_IT05}. Notice that the mobility model determines the distributions of nodes' locations and thus the opportunity that nodes encounter with each other, so it will affect the throughput capacity that can be achieved, as to be shown in Section~\ref{section:Markov} and Section~\ref{section:case_study}.

We consider the permutation traffic model widely used in previous literature \cite{Ciullo_TON11}, where $n$ distinct unicast traffic flows exist in the network, each node is the source of one flow and meanwhile the destination of another flow. We assume that the exogenous (self-generated) packet arrival at each node follows the i.i.d Bernoulli process with mean rate $\lambda$ packets/slot.

\begin{figure}[!t]
\centering\includegraphics[width=3.3in]{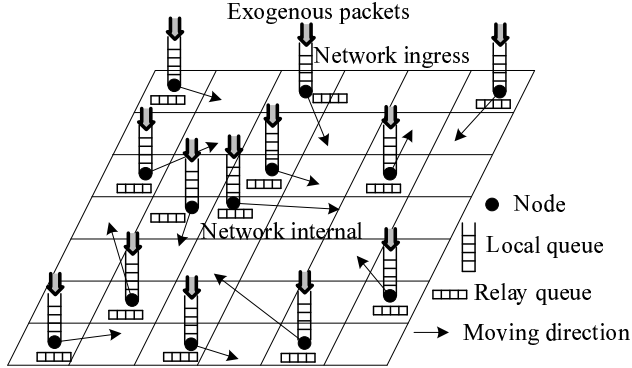}
\caption{System models.}
\label{fig:system_models}
\end{figure}

\subsection{Node Model and Buffer Constraint}  \label{subsection:node_model}

As illustrated in Fig.~\ref{fig:system_models} that we consider a practical node model similar to that of previous studies on buffer-limited wireless networks \cite{Herdtner_INFOCOM05,Le_TON12}, where each node maintains two independent queues, one local queue with unlimited buffer size for storing the exogenous packets of its own flow and one shared relay queue with fixed size $B$ for storing the relay packets coming from all other $n-2$ traffic flows. The local queue follows the FIFO (first-in-first-out) discipline. The relay queue follows the quasi-FIFO (Q-FIFO) discipline, i.e., the packets destined to the same node in the relay queue follow the FIFO discipline.

We adopt this node model here mainly due to the following reasons. First, in a practical network, each node usually reserves a much larger buffer space for storing its exogenous packets rather than the relay packets. Second, even though the local buffer space is not enough when bursty traffic comes, the upper layer (like transport layer) can execute congestion control to avoid the loss of local packets. Third, as illustrated in Fig.~\ref{fig:system_models} that the local queue serves as an ingress for exogenous packets entering a MANET, so it is actually an external module of the MANET and its buffer size does not affect the maximal achievable throughput the MANET can support. Finally, the relay queue serves as an important internal module of a MANET for storing and forwarding relay packets, so its buffer size has a critical impact on the network throughput performance.

\subsection{Routing Scheme}  \label{subsection:routing_scheme}

\begin{figure}[!t]
\centering\includegraphics[width=3.3in]{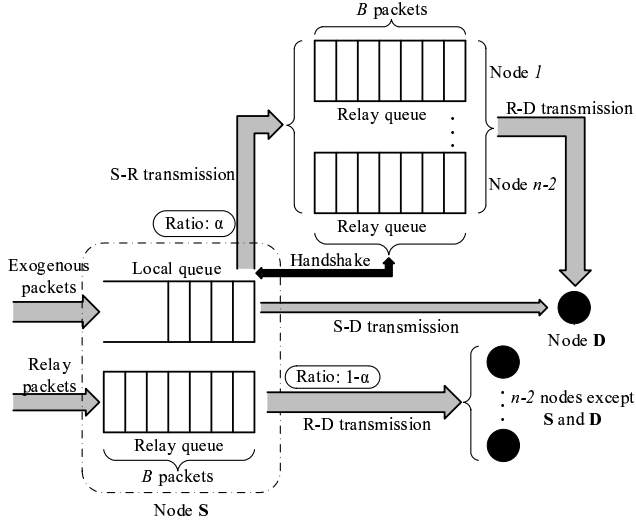}
\caption{Illustration of the 2HR-$\alpha$ routing scheme.}
\label{fig:routing_scheme}
\end{figure}

To support the efficient operation of buffer-limited MANETs, we extend the Two Hop Relay (2HR) algorithm \cite{Grossglauser_Tse_2001} and consider a more flexible Two Hop Relay scheme with parameter $\alpha$ (2HR-$\alpha$ for short), which incorporates both a control parameter $\alpha$ for transmission control and a handshaking mechanism for avoiding unnecessary packet loss. 

Regarding the 2HR scheme, for a tagged flow $(\mathbf{S},\mathbf{D})$ with source node $\mathbf{S}$ and destination node $\mathbf{D}$, when $\mathbf{S}$ gets access to the wireless channel in a time slot, it will transmit a packet directly to $\mathbf{D}$ (S-D transmission) if $\mathbf{D}$ is within its transmission range; otherwise with probability $0.5$, $\mathbf{S}$ selects to transmit a self-generated packet to a relay node (S-R transmission), or deliver a packet of other nodes to the corresponding destination (R-D transmission).

Notice that under the buffer-limited scenario, when node $\mathbf{S}$ executes the S-R transmission while the relay queue of the receiver is full, then the transmission will not be successful and the transmitted packet will be lost. To facilitate the operation of buffer-limited MANETs and improve the throughput performance, we adopt here the 2HR-$\alpha$ scheme, which is an extension of the 2HR scheme in the following two aspects. First, as illustrated in Fig.~\ref{fig:routing_scheme} that we introduce a parameter $\alpha$ to flexibly control the probability that $\mathbf{S}$ selects to conduct S-R transmission, i.e., when $\mathbf{S}$ gets access to the wireless channel and node $\mathbf{D}$ is not within its transmission range, $\mathbf{S}$ selects to transmit a self-generated packet to a relay with probability $\alpha$, and deliver a packet of other nodes to the corresponding destination with probability $1-\alpha$.  Thus, $\alpha$ represents the level of selfishness of a node, from $0$ (fully selfless) to $1$ (fully selfish), and it is expected that $\alpha$ should be set appropriately according to the network settings to achieve the optimal throughput performance. Second, to avoid the unnecessary packet loss in S-R transmission, the 2HR-$\alpha$ scheme further adopts a handshake mechanism to confirm the relay-buffer occupancy state of a receiver, where the S-R transmission will be conducted only when the relay queue of the intended receiver is not full.

\subsection{Basic Definitions} \label{subsection:definitions}

\textbf{Relay-buffer Blocking Probability (RBP)}: For a buffer-limited MANET with the 2HR-$\alpha$ scheme and a given exogenous packet arrival rate $\lambda$ to each node, the \emph{relay-buffer blocking probability} $p_b(\lambda)$ of a node is defined as the probability that the relay queue of this node is full. 

\textbf{Throughput}: The \emph{throughput} of a flow is defined as the time average of number of packets that can be delivered from its source to its destination. 

\textbf{Throughput Capacity and Optimal Throughput Capacity}: For a buffer-limited MANET with the 2HR-$\alpha$ scheme, its \emph{throughput capacity} $T_c$ is defined as the maximal achievable per-flow throughput the network can support, independent of the arrival rate. The \emph{optimal throughput capacity} $T_c^*$ is defined as the maximum value of throughput capacity $T_c$ optimized over the control parameter $\alpha$, i.e., $T_c^*=\mathop {\max }\limits_{\alpha \in [0,1]}T_c$.

\section{Throughput Capacity Analysis} \label{section:throughput_capacity}
For a buffer-limited MANET with the 2HR-$\alpha$ scheme, we denote by $p_{sd}$, $p_{sr}$ and $p_{rd}$ the probabilities that in a time slot a node gets access to the wireless channel and selects to execute S-D, S-R and R-D transmission respectively. With the help of these basic probabilities, we can establish the following theorem regarding the throughput capacity of the network. 

\begin{theorem} \label{theorem:throughput_capacity}
For a buffer-limited MANET with the 2HR-$\alpha$ scheme, its throughput capacity $T_c$ is determined as 
\begin{equation}
T_c=p_{sd}+p_{sr}(1-p_b(\tilde{\lambda})) \hspace{1cm} packets/slot, \label{eq:general_throughput_capacity}
\end{equation}
where $\tilde{\lambda}$ is the unique solution of  the following equation 
\begin{equation}
\lambda=p_{sd}+p_{sr}(1-p_b(\lambda)). \label{eq:lambda_equal_mu}
\end{equation}
\end{theorem}

\begin{proof}
To prove the Theorem, we first demonstrate that there exists an unique solution $\tilde{\lambda}$ for the equation~(\ref{eq:lambda_equal_mu}), and then show that the throughput is $\lambda$ when packet arrival rate $\lambda<\tilde{\lambda}$,  but the throughput is always $\tilde{\lambda}$ when $\lambda \geq \tilde{\lambda}$.

From our system models, it is clear that each flow experiences the same service process without priority, so the queuing process of each flow is identical and we can focus on a tagged flow in our analysis. Under the 2HR-$\alpha$ routing scheme, the delivering process of a packet from its source to destination involves at most two stages. The first stage is the queuing process at its source node, while the second stage is the queuing process at one relay node if the packet is not directly delivered to its destination. 

Concerning the first stage queuing process, the local queue there can be modeled as a
Bernoulli/Bernoulli queue with arrival rate $\lambda$ and service rate $\mu_s(\lambda)$ determined as
\begin{equation}
\mu_s(\lambda)=p_{sd}+p_{sr}\left(1-p_b(\lambda)\right).  \label{eq:mu_s}
\end{equation}

We can easily see that: 1) when $\lambda=0$, we have $p_b(0)=0$, so $\mu_s(0)=p_{sd}+p_{sr}>\lambda$; 2) as $\lambda$ increases, $p_b(\lambda)$ tends to increase, leading to a decrease in  $\mu_s(\lambda)$; 3) when $\lambda=p_{sd}+p_{sr}$, we have $p_b(\lambda)>0$, $\mu_s(\lambda)<\lambda$. Based on these properties of service rate $\mu_s(\lambda)$, we know that there exists an unique $0<\tilde{\lambda}<p_{sd}+p_{sr}$ such that $\tilde{\lambda}=\mu_s(\tilde{\lambda})$.

Considering a time interval $[0,t]$, we denote by $m_0(t)$ and $m_1(t)$ the number of packets being buffered in all local queues and all relay queues at time slot $t$, respectively.  Since the total number of exogenous arrival packets during this interval is $n \lambda t$, then the throughput $Th$ is determined as
\begin{equation} 
Th=\lim\limits_{t\to\infty}\frac{n \lambda t-m_0(t)-m_1(t)}{n\cdot t}. \label{eq:throughput}
\end{equation}

Since the relay-buffer of each node has a fixed size $B$, then $\frac{m_1(t)}{n}\leq B$ and $\lim\limits_{t\to\infty}\frac{m_1(t)}{n\cdot t}=0$. 

For the case $\lambda<\tilde{\lambda}$, we denote by $L_s$ the queue length of the local queue, then the expectation $\mathbb{E}\{L_s\}$ of $L_s$ is given by \cite{Daduna_BOOK01}
\begin{equation}
\mathbb{E}\{L_s\}=\frac{\lambda-\lambda^{2}}{\mu_s(\lambda)-\lambda}. \label{eq:local_queue_length}
\end{equation}
Since when $\lambda<\tilde{\lambda}$, we have $\mu_s(\lambda)>\lambda$, so the queue length $\mathbb{E}\{L_s\}$ is bounded in this case. Thus, $\lim\limits_{t\to\infty}\frac{m_0(t)}{n\cdot t}=0$ and $Th=\lambda$.

When $\lambda \geq \tilde{\lambda}$, then $\mu_s(\lambda)<\lambda$, leading to an increasing number of packets buffered in the local queues. By applying the law of large numbers, we have that as $t \to \infty$
\begin{equation}
\frac{m_0(t)}{t} \overset{a.s.}{\to} n (\lambda - \mu_s(\tilde{\lambda})). \nonumber
\end{equation}
Based on (\ref{eq:throughput}), we then have $Th=\tilde{\lambda}$ when $\lambda \geq \tilde{\lambda}$.

Thus, the throughput capacity $T_c$ of the concerned network is determined as
\begin{equation}
T_c=\mu_s(\tilde{\lambda})=p_{sd}+p_{sr}\left(1-p_b(\tilde{\lambda})\right).  \nonumber
\end{equation}
\end{proof}

\begin{remark}
Notice that for the heterogeneous network scenario, the network level capacity region is defined by a vector $\mathbf{\Lambda}=\{\lambda_1,\lambda_2,\cdots,\lambda_n\}$, where $\lambda_i$ denotes the maximum packet arrival rate of node $i$ that the network can stably support. For the homogeneous network scenario considered in this paper, the behavior of each node is the same, so we have $\lambda_1=\lambda_2=\cdots=\lambda_n$. Thus, the network level capacity here reduces to the per node capacity as expressed in (\ref{eq:general_throughput_capacity}).  
\end{remark}

\section{Embedded Markov Chain Framework} \label{section:Markov}
The result in Theorem~\ref{theorem:throughput_capacity} indicates that for the throughput capacity analysis of the concerned MANET, we need to determine RBP $p_b(\lambda)$ in such network. To address this issue, in this section we first utilize a two-dimensional Markov Chain (MC) to depict the complicated state transitions of a relay queue, then convert the two-dimensional MC into a new Embedded Markov Chain (EMC) to obtain its one-step transition probability, such that a complete theoretical framework is developed to enable the RBP and thus the exact throughput capacity of the concerned MANET to be derived.  

\subsection{State Machine of Relay Queue}
\begin{figure}[!t]
\centering
{
\subfloat[Transitions of a general state $(i,k)$.]{\includegraphics[width=2.5in]{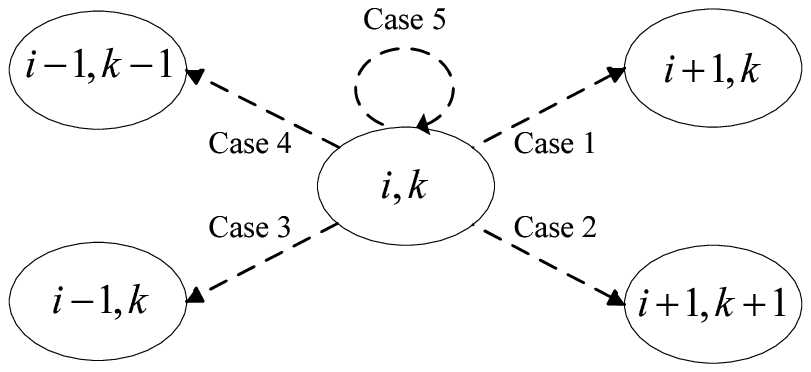}\label{fig:transition_cases}}
\hspace{0.15in}
\subfloat[State machine of a relay queue.]{\includegraphics[width=3.0in]{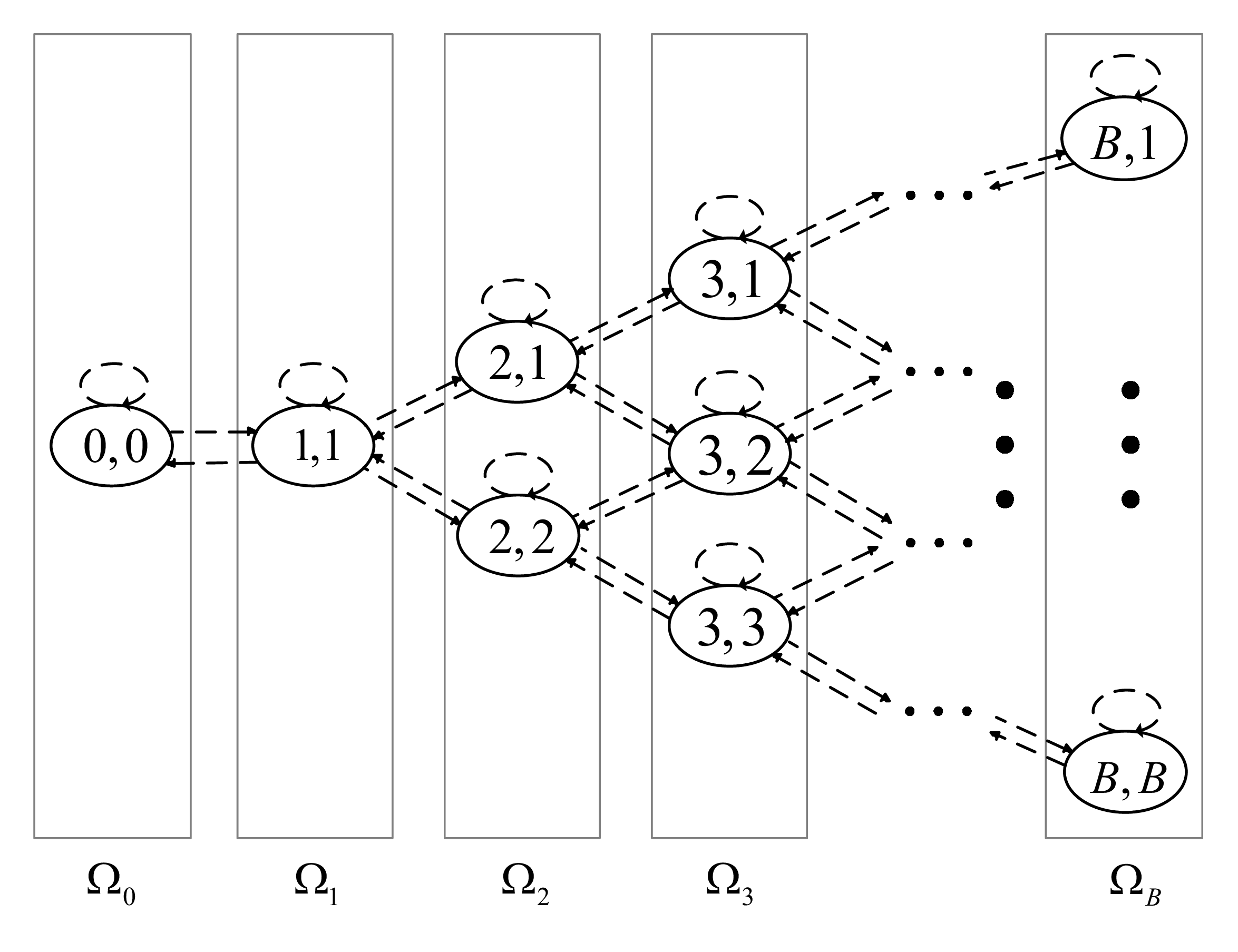}\label{fig:state_machine}}
}
\caption{Transitions and state machine of a relay queue.}
\label{fig:transition_state}
\end{figure}

Without loss of generality, we focus on the relay queue of a general node and use a two-tuple $(i,k)$ to
denote the state that the relay queue contains $i$ packets and these packets are destined for $k$ distinct destination nodes, here $0 \leq i \leq B$, $k=0$ when $i = 0$ and $1 \leq k \leq i$ when $i>0$.
As illustrated in Fig.~\ref{fig:transition_cases} that when the relay queue is in state $(i,k)$ at current time slot, only one of the following transitions may happen in the next time slot:
 
\begin{itemize}

\item{Case 1:}
transition from $(i,k)$ to $(i+1,k)$, i.e., a packet enters the relay queue and this packet is destined for a destination same as one of packet(s) already in the queue. 

\item{Case 2:}
transition from $(i,k)$ to $(i+1,k+1)$, i.e., a packet enters the relay queue and the destination of this packet is different from all packet(s) in the queue. 

\item{Case 3:}
transition from $(i,k)$ to $(i-1,k)$, i.e., a packet from the relay queue is delivered to its destination, but there still exist other packet(s) in the relay queue destined for this destination. 

\item{Case 4:}
transition from $(i,k)$ to $(i-1,k-1)$, i.e., a packet from the relay queue is delivered to its destination, and none of the remaining packet(s) in the relay queue is destined for this destination. 

\item{Case 5:}
transition from $(i,k)$ to $(i,k)$, i.e., no packet entering into or departing from the relay queue. 

\end{itemize}

Based on above transitions, the state machine of the relay queue can then be modeled as a discrete-time finite-state Markov chain illustrated in Fig.~\ref{fig:state_machine}, where  $\Omega_i$ denotes the set of states $(i,k), k\leq i$.

\subsection{Markov Chain Collapsing} \label{chain_collapsing}
It is easy to see that the total number of states of the Markov chain in Fig.~\ref{fig:state_machine} is $1+\frac{(1+B)B}{2}$, so the Markov chain will become too complicated to be solved when $B$ is big. Notice also that for the analysis of RBP $p_b(\lambda)$, the limiting distribution on each set $\Omega_i$ rather than the limiting distribution on each state $(i,k)$ is of concern. For these reasons, we apply the novel ``Markov Chain Collapsing'' \cite{Hachigian_1963} technique to convert the Markov chain in Fig.~\ref{fig:state_machine} into a new EMC illustrated in Fig.~\ref{fig:embedded_markov}, where $\Omega_i$ denotes the state that the relay queue contains $i$ packets and $p_{\Omega_i,\Omega_j}$ denotes the one-step transition probability from $\Omega_i$ to $\Omega_j$. 

With the help of the EMC model in Fig.~\ref{fig:embedded_markov} and also the Markov Chain model in Fig.~\ref{fig:state_machine}, we can establish the following lemmas regarding the evaluation of $p_{\Omega_i,\Omega_j}$ and also the limiting distribution of the EMC. 

\begin{figure}[t]
\centering\includegraphics[width=3.3in]{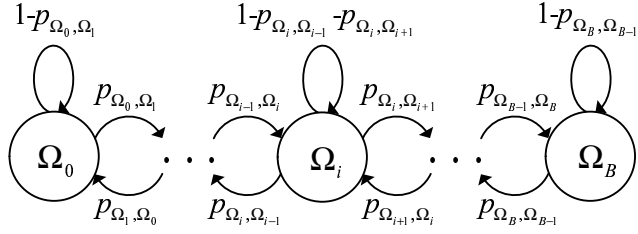}
\caption{State machine of the EMC.}
\label{fig:embedded_markov}
\end{figure}

\begin{lemma} \label{lemma:transition}
The one-step transition probability $p_{\Omega_i,\Omega_j}$ of the EMC in Fig.~\ref{fig:embedded_markov} is determined as
\begin{equation}
p_{\Omega_i,\Omega_j}=\left\{ 
\begin{aligned}
&\rho_s(\lambda) \cdot p_{sr}, &   &j=i+1 \leq B \\
&\frac{i}{n-3+i} \cdot p_{rd}, &   &j=i-1 \geq 0 \\
&1-p_{\Omega_i,\Omega_{i+1}}-p_{\Omega_i,\Omega_{i-1}}, & & j=i \\
&0,                            &   &\text{others} 
\end{aligned}
\right.
\label{eq:transition}
\end{equation}
where $\rho_s(\lambda)=\frac{\lambda}{p_{sd}+p_{sr}(1-p_b(\lambda))}=\frac{\lambda}{\mu_s(\lambda)}$.
\end{lemma}

\begin{proof}
See \ref{appendix:transition} for the proof.
\end{proof}

\begin{lemma} \label{lemma:limit_distribution}
The limiting distribution $\mathbf{\Pi}=\left\{\pi(\Omega_0),\pi(\Omega_1),\cdots,\pi(\Omega_B)\right\}$ of the EMC exists and is unique, and it is given by
\begin{align}
\pi(\Omega_0)&=\frac{1}{\sum_{i=0}^B{\mathrm{C}_i \cdot \beta^i \cdot \rho_s(\lambda)^i}}, \label{eq:pi_0} \\
\pi(\Omega_i)&=\frac{\mathrm{C}_i \cdot \beta^i \cdot \rho_s(\lambda)^i} {\sum_{i=0}^B{\mathrm{C}_i \cdot \beta^i \cdot \rho_s(\lambda)^i}},  \quad 0<i\leq B  \label{eq:pi_i}
\end{align}
where $\mathrm{C}_i=\binom{n-3+i}{i}$ and $\beta=\frac{p_{sr}}{p_{rd}}=\frac{\alpha}{1-\alpha}$. 
\end{lemma}

\begin{proof}
See \ref{appendix:limit_distribution} for the proof.
\end{proof}

\subsection{Derivation of $p_b(\lambda)$ and $T_c$}
Based on above EMC-based framework, we now provide analysis on the RBP $p_b(\lambda)$ and the exact throughput capacity $T_c$, as summarized in following theorem.

\begin{theorem} \label{theorem:T_c}
(\textbf{Main result}) For a concerned MANET with $n$ mobile nodes, where each node is allocated with a relay-buffer of fixed size $B$ and the 2HR-$\alpha$ scheme is adopted for packet delivery, its RBP is determined by the following equation
\begin{equation}
p_b(\lambda)=\frac{\mathrm{C}_B \cdot \beta^B \cdot \rho_s(\lambda)^B} {\sum_{i=0}^B{\mathrm{C}_i \cdot \beta^i \cdot \rho_s(\lambda)^i}}, \label{eq:blocking_pro}
\end{equation}
and its throughput capacity $T_c$ is determined as
\begin{equation}
T_c=p_{sd}+p_{sr}\left( 1-\frac{\mathrm{C}_B \cdot \beta^B}{\sum_{i=0}^B {\mathrm{C}_i \cdot \beta^i}}\right).
\label{eq:throughput_capacity}
\end{equation}

\end{theorem}

\begin{proof}
Notice that when the EMC is on state $\Omega_B$, the relay queue is full. It indicates that RBP $p_b(\lambda)$ is just equal to the limiting distribution $\pi(\Omega_B)$, so (\ref{eq:blocking_pro}) follows. Since $\rho_s(\lambda)=\frac{\lambda}{p_{sd}+p_{sr}(1-p_b(\lambda))}$, equation~(\ref{eq:blocking_pro}) contains only one unknown quantity $p_b(\lambda)$. By solving equation~(\ref{eq:blocking_pro}), we can then obtain the RBP $p_b(\lambda)$ under any exogenous packet arrival rate $\lambda$.

From Theorem~\ref{theorem:throughput_capacity} we know that as $\lambda$ approaches $\tilde{\lambda}$, $\rho_s(\lambda)$ tends to $1$. Substituting $\rho_s(\tilde{\lambda})=1$ into (\ref{eq:blocking_pro}), we have
\begin{equation}
p_b(\tilde{\lambda})=\frac{\mathrm{C}_B \cdot \beta^B} {\sum_{i=0}^B{\mathrm{C}_i \cdot \beta^i}}. \label{eq:pb_tc}
\end{equation} 
The formula (\ref{eq:throughput_capacity}) then follows by substituting (\ref{eq:pb_tc}) into (\ref{eq:general_throughput_capacity}).

\end{proof}

\begin{remark} \label{remark:applicable}
It is notable that our EMC-based framework for throughput capacity analysis is very general in the sense that it can be applied to other network models (like the continuous network model \cite{Chen_2013_Aloha}) and other mobility models (like the random walk mobility model), as long as they lead to the same steady distributions of the nodes' locations as the i.i.d mobility model.
\end{remark}

Based on Theorem~\ref{theorem:T_c}, we have the following corollaries (see \ref{appendix:corollaries} for the proofs).

\begin{corollary} \label{corollary:B_grow}
For a network with $n \geq 3$, its throughput capacity $T_c$ increases as relay-buffer size $B$ grows.
\end{corollary}

\begin{corollary} \label{corollary:tc_alpha_0.5}
With the setting of $\alpha=0.5$, i.e., each node executes S-R and R-D transmission with equal probability, $T_c$ is determined as
\begin{equation}
T_c=p_{sd}+p_{sr}\frac{B}{n-2+B}
\label{eq:tc_alpha_0.5}
\end{equation}
\end{corollary}

\begin{corollary} \label{corollary:tc_infinity}
When the relay-buffer size $B$ tends to infinity, the throughput capacity $T_c$ is determined as

\begin{equation}
T_c|_{B\to \infty}=\left\{ 
\begin{aligned}
&p_{sd}+p_{sr}, &   &\alpha \leq 0.5 \\
&p_{sd}+p_{rd}, &   &\alpha >0.5 
\end{aligned}
\right.
\label{eq:tc_infinity}
\end{equation}
\end{corollary}


\section{Case Studies} \label{section:case_study}
The results in Theorem~\ref{theorem:T_c} indicate for the evaluation of throughput capacity of a MANET concerned in this paper, we need to determine the probabilities $p_{sd}$ and $p_{sr}$ in the network, which are further related to the schemes adopted for transmission scheduling. To demonstrate the applicability of our Embedded Markov Chain-based framework for throughput capacity analysis, this section provides case studies under two typical transmission scheduling schemes, i.e., local transmission scheduling (LTS) and group-based transmission scheduling (GTS). The related issues of capacity optimization and scaling law performance will be also explored.

\subsection{Throughput Capacity under LTS}
We first provide analysis of throughput capacity under the LTS scheme \cite{Neely_IT05}, where two nodes within the same cell can forward one packet during a time slot, and nodes within different cells cannot communicate. Same as \cite{Neely_IT05}, the node/cell density $d$ is assumed to be $O(1)$ (independent of $n$) and without loss of generality, $n$ is assumed to be even and source-destination pairs are composed as: $1 \leftrightarrow 2, 3 \leftrightarrow 4, \cdots, n-1 \leftrightarrow n$.

We denote by $p_0$ and $p_1$ the probabilities that there are at least two nodes in a cell and there is at least one source-destination pair in a cell, respectively. Based on the results of \cite{Neely_IT05}, we then have
\begin{align}
p_0&=1-\left(1-\frac{1}{m^2}\right)^{n}-\frac{n}{m^2}\left(1-\frac{1}{m^2}\right)^{n-1},  \label{eq:p_0}\\
p_1&=1-\left(1-\frac{1}{m^{4}}\right)^{n/2}.   \label{eq:p_1}
\end{align}

At a time slot, the total transmission opportunity in the network is $m^2 \cdot p_0$, which is shared equally by all nodes, so we have

\begin{equation}
n\cdot(p_{sd}+p_{sr}+p_{rd})=m^2\cdot p_0. \label{eq:total_opportunity}
\end{equation}
Similarly,
\begin{equation}
n\cdot p_{sd}=m^2\cdot p_1. \label{eq:direct_opportunity}
\end{equation}
Combining with $\frac{p_{sr}}{p_{rd}}=\frac{\alpha}{1-\alpha}$, we have
\begin{align} 
p_{sd}&=\frac{1}{d}p_1,   \label{eq:p_sd}   \\
p_{sr}&=\frac{\alpha}{d}(p_0-p_1),  \label{eq:p_sr} \\
p_{rd}&=\frac{1-\alpha}{d}(p_0-p_1).  \label{eq:p_rd}
\end{align}

Substituting (\ref{eq:p_sd}) and (\ref{eq:p_sr}) into (\ref{eq:throughput_capacity}), we can see that the throughput capacity under the LTS scheme is determined as 
\begin{equation}
T_c=\frac{1}{d}p_1+\frac{\alpha}{d}(p_0-p_1) \left( 1-\frac{\mathrm{C}_B \cdot \beta^B}{\sum_{i=0}^B {\mathrm{C}_i \cdot \beta^i}}\right). \label{eq:tc_LTS}
\end{equation}

\begin{remark}
From Corollary~\ref{corollary:tc_infinity} we can see that when  $\alpha=0.5$ and $B \to \infty$, then
(\ref{eq:tc_LTS}) is reduced to the capacity result in \cite{Neely_IT05}, i.e., $T_c=\frac{p_0+p_1}{2d}$. 
\end{remark}

Regarding the optimal throughput capacity $T_c^*$ and the corresponding optimal setting of $\alpha^*$  under the LTS scheme, we have the following theorem.

\begin{theorem} \label{theorem:optimization}
For a concerned MANET with the 2HR-$\alpha$ scheme and a fixed relay-buffer size $B$, its optimal throughput capacity $T_c^*$ under the LTS scheme is determined as
\begin{equation}
T_c^*=\frac{1}{d}p_1+\frac{p_0-p_1}{d(1+\gamma^*)}\frac{h(\gamma^*)}{h(\gamma^*)+\mathrm{C}_B}, \label{eq:tc^*}
\end{equation}
and the corresponding optimal transmission ratio $\alpha^*$ is given by $\alpha^*=\frac{1}{1+\gamma^*}$, where 
\begin{equation}
h(\gamma)=\sum_{i=0}^{B-1} {\mathrm{C}_i \cdot \gamma^{B-i}}, \label{eq:h(gamma)}
\end{equation} 
$h'(\gamma)$ is the derivative of $h(\gamma)$, and $r^*$ is determined by the following equation
\begin{equation}
h(\gamma^*)[h(\gamma^*)+\mathrm{C}_B]=(1+\gamma^*)\mathrm{C}_B h'(\gamma^*).
\label{eq:gamma^*}
\end{equation}

\end{theorem}

\begin{proof}
We define $\gamma=\frac{1-\alpha}{\alpha}$ (i.e., $\alpha=\frac{1}{1+\gamma}$, $\beta=\frac{1}{\gamma}$), and $g(\gamma)=(1+\gamma)\left(1+\frac{\mathrm{C}_B}{h(\gamma)} \right)$. From (\ref{eq:tc_LTS}) we can see that the optimal throughput capacity $T_c^*$ is determined as
\begin{align}
T_c^*&=\mathop {\max }\limits_{\alpha \in [0,1]}T_c \nonumber \\
&=\frac{1}{d}p_1+\frac{1}{d}(p_0-p_1) \cdot \mathop {\max }\limits_{\alpha \in [0,1]} \left \{ \alpha \left( 1-\frac{\mathrm{C}_B \cdot \beta^B}{\sum_{i=0}^B {\mathrm{C}_i \cdot \beta^i}}\right) \right \} \nonumber \\
&= \frac{1}{d}p_1+\frac{p_0-p_1}{d} \frac{1}{\mathop {\min }\limits_{\gamma \geq 0}g(\gamma)}. \label{eq:optimization}
\end{align}

We can see that: 1) $g(\gamma)$ is an Elementary Function \cite{Muller_BOOK1997}, so it is continuous and differentiable on the interval $\gamma \in (0,\infty)$; 2) $\mathop {\lim }\limits_{\gamma \to 0}g(\gamma)\to \infty$, $\mathop {\lim }\limits_{\gamma \to \infty}g(\gamma) \to \infty$ and $g(\gamma)>0$. According to the Extreme Value Theorem \cite{Keisler_BOOK2012}, there must exists $0<\gamma^*<\infty$ such that $0<g(\gamma^*) \leq g(\gamma)$ for $\forall \gamma \in (0,\infty)$ and $g'(\gamma^*)=0$, so equation~(\ref{eq:gamma^*}) follows. Then formula (\ref{eq:tc^*}) follows by substituting $\gamma^*$ into (\ref{eq:optimization}).

\end{proof}

Based on Theorem~\ref{theorem:optimization} we have the following corollary.
\begin{corollary} \label{corollary:optimal_alpha}
For any setting of $n$ and $B$, $\alpha^*<0.5$; when $B \to \infty$, $\alpha^*|_{B \to \infty} = 0.5$ and $T_c^*|_{B \to \infty} = \frac{p_0+p_1}{2d}$.
\end{corollary}

\begin{proof}
See \ref{appendix:optimal_alpha} for the proof.
\end{proof}

With the help of exact expression of throughput capacity (\ref{eq:tc_LTS}), the achievable scaling law under the LTS scheme can be further explored, as shown in the following corollary.

\begin{corollary} \label{corollary:scaling_law}
When the number of nodes in the concerned MANET tends to infinity, each node can achieve a throughput $T_c(n)$ as
\begin{equation}
T_c(n)=\Theta\left( \frac{B}{n}\right).
\label{eq:scaling_law}
\end{equation}
\end{corollary}

\begin{proof}
See \ref{appendix:scaling_law} for the proof.
\end{proof}

\begin{remark} \label{remark:scaling_law}
Notice that Corollary~\ref{corollary:scaling_law} is the first time to reveal an achievable order sense (in $\Theta$ form) on per node throughput of a buffer-limited MANET, not just an upper bound (in $O$ form) provided in \cite{Herdtner_INFOCOM05}, and it indicates that as $n$ increases, the concerned MANET can still achieve a non-vanishing throughput as long as its relay-buffer size $B$ grows at least linearly with $n$.
\end{remark}

\subsection{Throughput Capacity under GTS}
We further conduct analysis on throughput capacity under the GTS scheme \cite{Kulkarni_IT04, Ciullo_TON11, Gao_2013}. With GTS scheme, all cells are divided into different groups, where any two cells in the same group have a horizontal and vertical distance of some multiple of $\epsilon$ cells. Thus, the MANET has $\epsilon^2$ groups and each group contains $J=\lfloor m^2/\epsilon ^2 \rfloor$ cells. Each group becomes active every $\epsilon^2$ time slots and a node in an active cell can transmit one packet to another node within a horizontal and vertical distance of $\nu -1$ cells. Based on the Protocol Model \cite{Gupta_IT00}, to ensure that concurrent transmissions are not interfering with each other, $\epsilon$ is determined as \cite{Gao_2013}
\begin{equation}
\epsilon=\min \{\lceil (1+\Delta) \sqrt{2} \nu+\nu \rceil,m \}. \label{eq:epsilon}
\end{equation}

Considering a given time slot and a given active cell $c$, we denote by $p_3$ the probability that there are at least one node within $c$ and another node within the transmission range of $c$, and denote by $p_4$ the probability that there are at least one source-destination pair within the transmission range of $c$ and for each of such pair(s), at least one of its two nodes is within $c$. Based on the results of \cite{Gao_2013}, we have
\begin{align}
& p_3=\frac{1}{m^{2n}}[m^{2n}-(m^2-1)^n-n(m^2-l)^{n-1}], \label{eq:p3_GTS} \\
& p_4=\frac{1}{m^{2n}}[m^{2n}-(m^4-2l+1)^{n/2}], \label{eq:p4_GTS}
\end{align} 
where $l=(2\nu -1)^2$. Notice also that $p_{sd}=\frac{J}{n}p_4$, $p_{sr}=\frac{\alpha J}{n}(p_3-p_4)$ and $p_{rd}=\frac{(1-\alpha)J}{n}(p_3-p_4)$. By substituting these results into (\ref{eq:throughput_capacity}), the throughput capacity under the GTS scheme is then determined as
\begin{equation}
T_c=\frac{J}{n}p_4+\frac{\alpha J}{n}(p_3-p_4) \left( 1-\frac{\mathrm{C}_B \cdot \beta^B}{\sum_{i=0}^B {\mathrm{C}_i \cdot \beta^i}}\right). \label{eq:tc_GTS}
\end{equation}

We can see that when $\alpha=0.5$ and $B \to \infty$, then (\ref{eq:tc_GTS}) is reduced to the capacity result in \cite{Gao_2013}, i.e., $T_c=\frac{J(p_3+p_4)}{2n}$.
Based on the proofs similar to that of Theorem~\ref{theorem:optimization} and Corollary \ref{corollary:scaling_law}, we have the following corollary regarding the optimal throughput capacity $T_c^*$ and scaling law under the GTS scheme.
\begin{corollary} 
For a concerned MANET, its optimal throughput capacity $T_c^*$ under the GTS scheme is determined as
\begin{equation}
T_c^*=\frac{J}{n}p_1+\frac{J(p_3-p_4)}{n(1+\gamma^*)}\frac{h(\gamma^*)}{h(\gamma^*)+\mathrm{C}_B}, \label{eq:tc^*_GTS}
\end{equation}
and the corresponding optimal transmission ratio $\alpha^*$ is given by $\alpha^*=\frac{1}{1+\gamma^*}$, where $h(\gamma)$ and $\gamma^*$ are determined by (\ref{eq:h(gamma)}) and (\ref{eq:gamma^*}), respectively. The order of throughput capacity $T_c(n)$ under the GTS scheme is determined as
\begin{equation}
T_c(n)=\Theta\left( \frac{B}{n}\right).
\label{eq:scaling_law_GTS}
\end{equation}
\end{corollary}

\section{Simulation Results and Discussions} \label{section:simulation}

In this section, we first provide the simulation results to validate our theoretical framework for the throughput capacity analysis of buffer-limited MANETs, and then apply our theoretical results to illustrate the performance of such networks.

\subsection{Simulation Setting} \label{subsection:simulation_setting}
For the validation of our framework, a C++ simulator was developed to simulate the packet delivery process in the concerned MANETs under both LTS and GTS schemes \cite{C++}. In addition to the i.i.d mobility model, the random walk model was also implemented in the simulator. Under the random walk model, at the beginning of each time slot, every node independently selects a cell among its current cell and its $8$ adjacent cells with equal probability $1/9$, then stays in it until the end of this time slot \cite{Gamal_IT06}.

\begin{figure}[!t]
    \centering
    {
    \subfloat[Throughput performance under the LTS scheme.]
    {\includegraphics[width=3.0in]{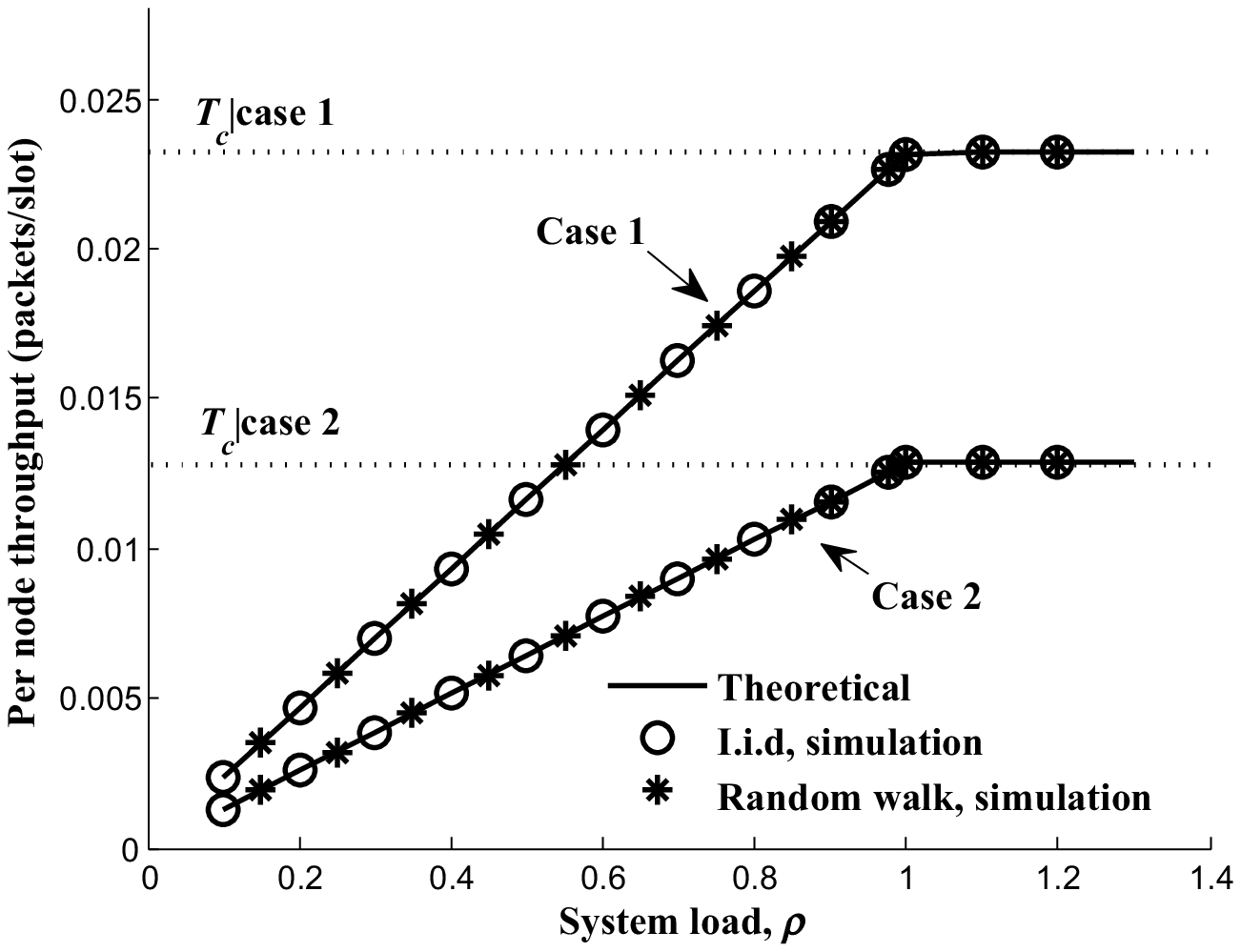} \label{subfig:throughput_LTS}}
   	\hspace{0.15in}
    \subfloat[Throughput performance under the GTS scheme.]
    {\includegraphics[width=3.0in]{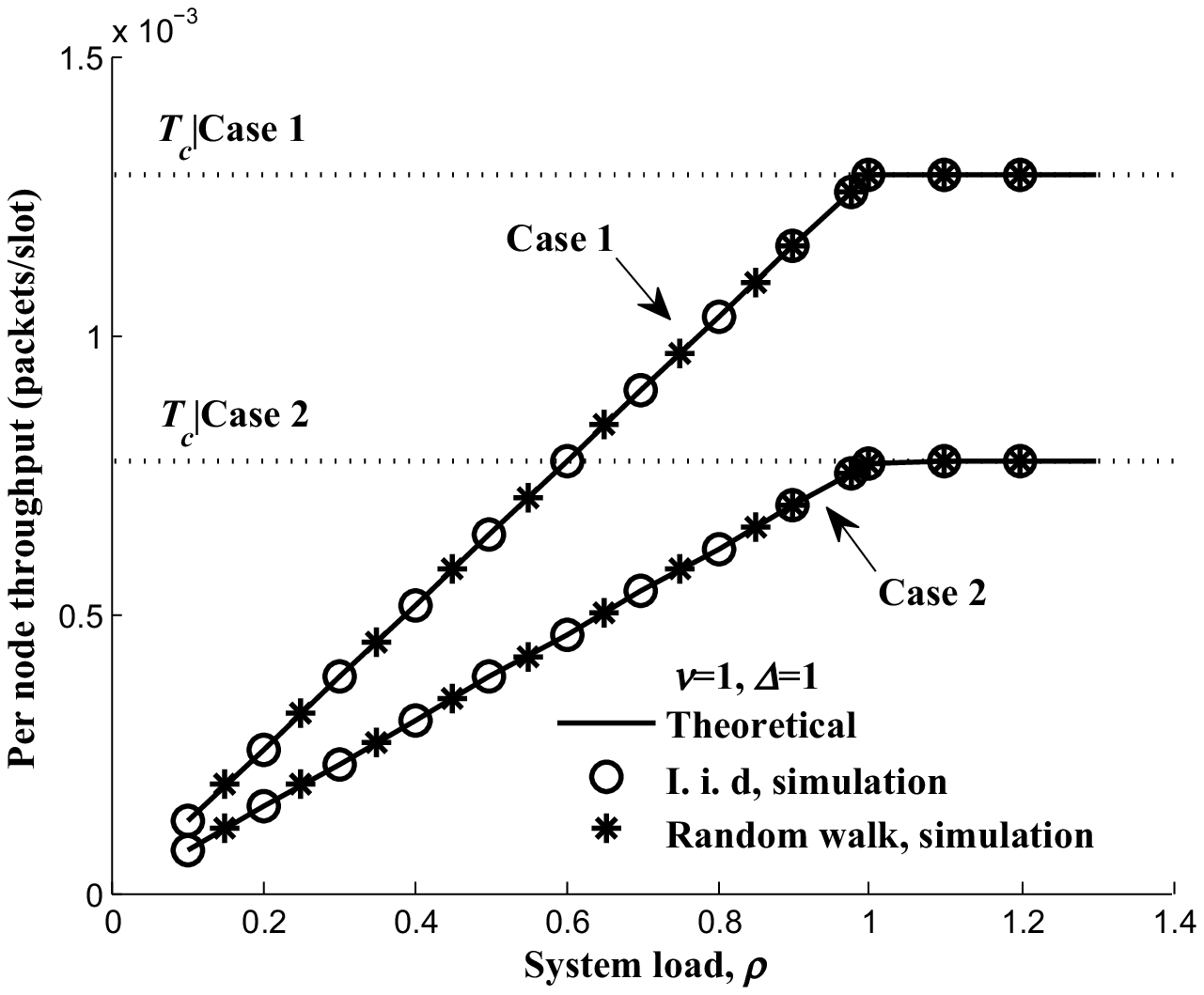} \label{subfig:throughput_GTS}}
    }
    \caption{Throughput performance under the LTS and GTS schemes. Case 1: $n=72, m=6, B=5, \alpha=0.5$. Case 2: $n=200, m=10, B=8, \alpha=0.3$.}
    \label{fig:throughput_simulation}
\end{figure}

Two network scenarios of ($n=72,m=6,B=5,\alpha=0.5$) and ($n=200,m=10,B=8,\alpha=0.3$) are considered in the simulation, where we set $\nu=1$ and $\Delta=1$ for the GTS scheme \cite{ns-2}. To simulate the throughput, we focus on a specific node and count its received packets over a period of $2 \times 10^8$ time slots, and then calculate the averaged number of packets this node can receive per time slot. The system load $\rho$ is defined as $\rho=\lambda/T_c$, and $T_c$ is given by (\ref{eq:tc_LTS}) and (\ref{eq:tc_GTS}) for the LTS and GTS, respectively. 

\subsection{Simulation Results} \label{subsection:simulation_results}

To validate the throughput capacity results (\ref{eq:tc_LTS}) and (\ref{eq:tc_GTS}), we provide plots of throughput versus system load $\rho$ in Fig.~\ref{fig:throughput_simulation}. It can be observed from Fig.~\ref{fig:throughput_simulation} that the simulation results agree well with the theoretical ones under both LTS and GTS schemes, indicating that our framework is highly efficient in capturing the throughput behaviors of concerned buffer-limited MANETs. We can see from Fig.~\ref{fig:throughput_simulation} that just as Theorem~\ref{theorem:throughput_capacity} predicates that for a concerned MANETs, its throughput increases linearly with $\rho$ when $\rho \leq 1$ and then keeps as a constant $T_c$ determined by (\ref{eq:general_throughput_capacity}) when $\rho>1$. A further observation of Fig.~\ref{fig:throughput_simulation} indicates that for a network under the random walk mobility model, its
throughput performance is very similar to that under the i.i.d mobility model. This is expected since according to Remark~\ref{remark:applicable} and \cite{Neely_IT05, Neely_05JSEC}, nodes in the network with i.i.d mobility model and random walk model have the same steady distribution, leading to the same throughput capacity performance under these two mobility models.

\subsection{Throughput Capacity} \label{subsection:throughput_capacity}
\begin{figure}[!t]
\centering\includegraphics[width=3.0in]{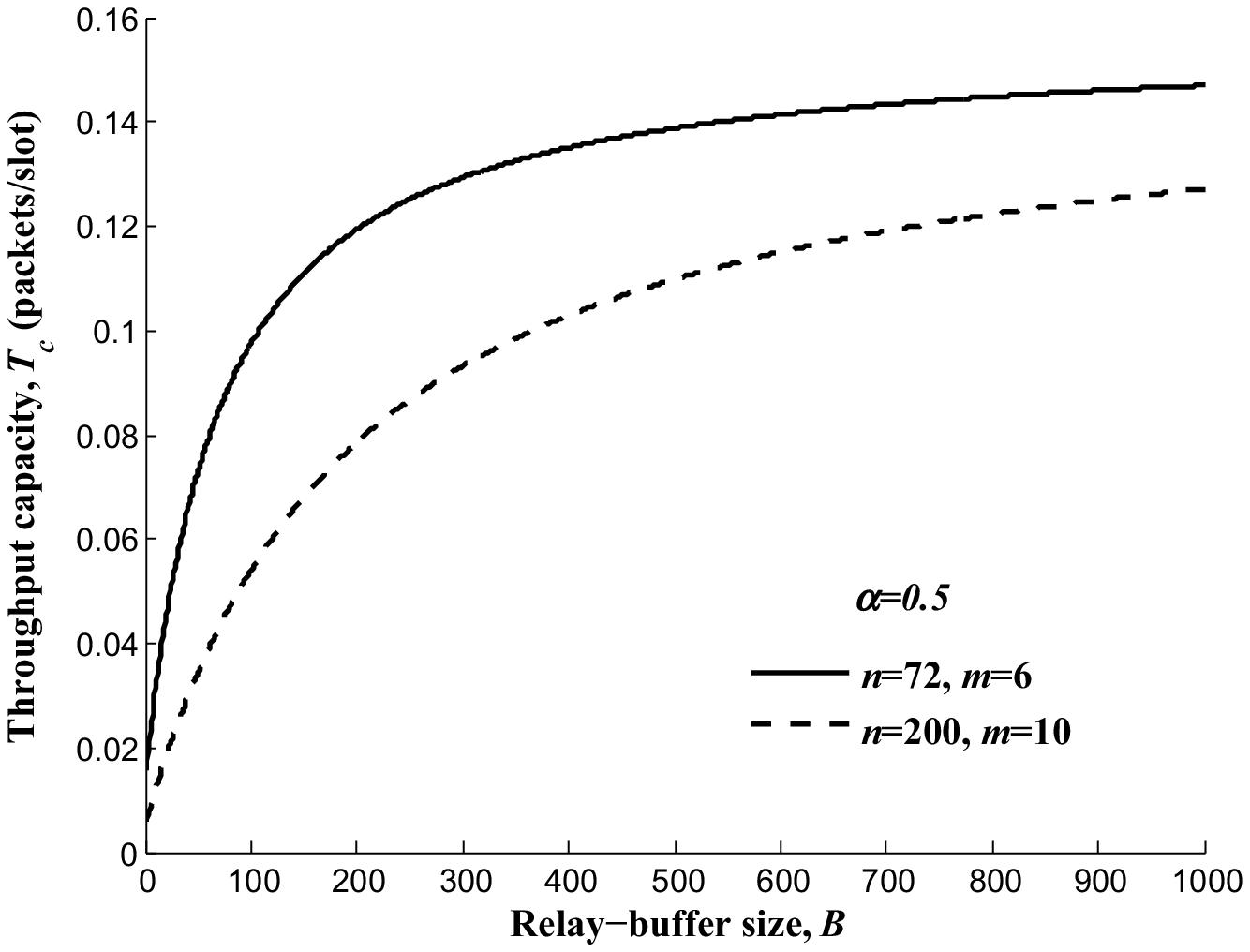}\caption{Throughput capacity $T_c$ vs. relay-buffer size $B$.}
\label{fig:tc_vs_B}
\end{figure}
\begin{figure}[!t]
\centering\includegraphics[width=3.0in]{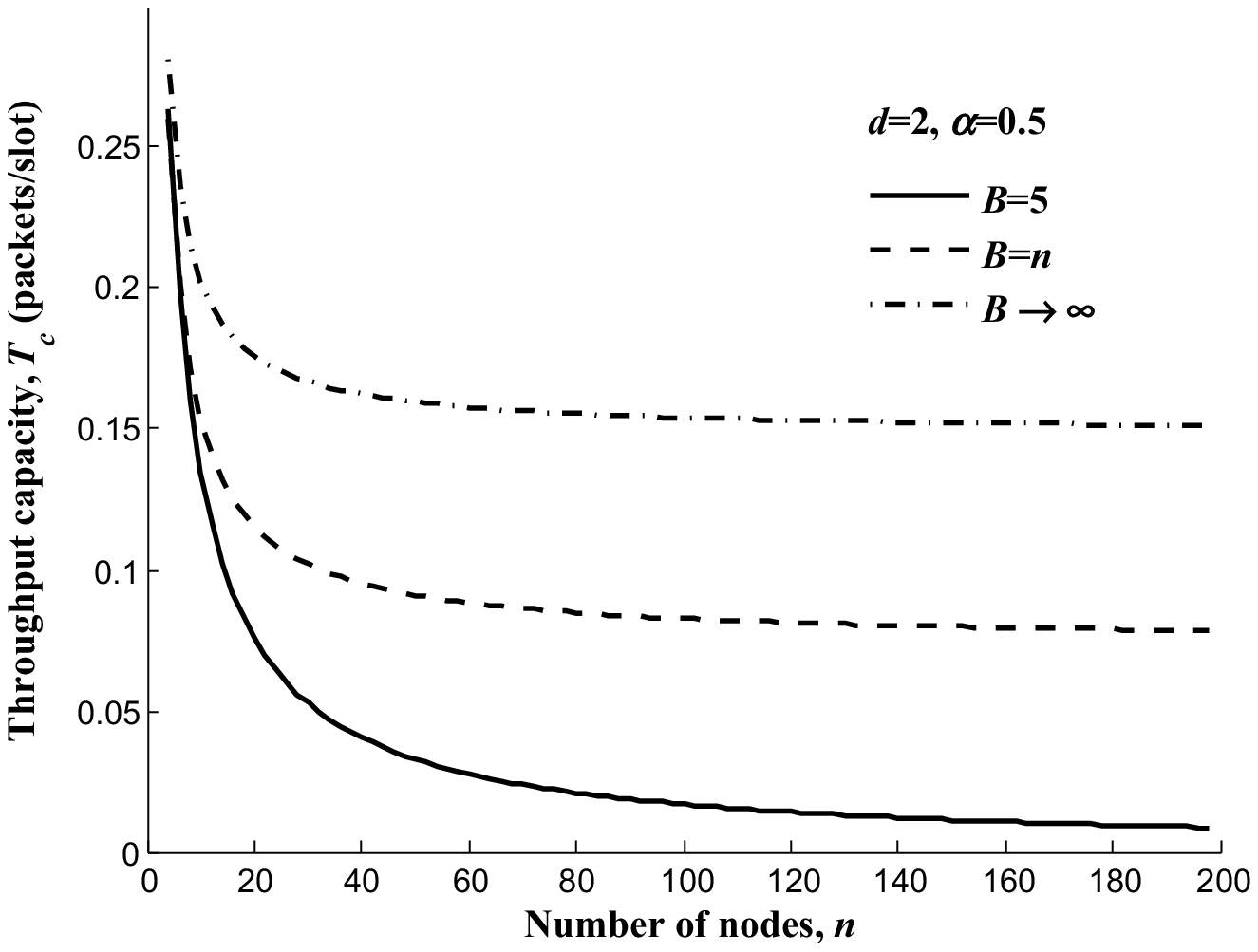}\caption{Throughput capacity $T_c$ vs. number of nodes $n$.}
\label{fig:tc_vs_n}
\end{figure}

With the help of our theoretical results, we illustrate here the impacts of network parameters on the throughput capacity. Notice that for a concerned MANET its overall throughput behavior under the LTS is very similar to that under the GTS, so we consider the LTS scheme here only for illustration.

We first summarize in Fig.~\ref{fig:tc_vs_B} how throughput capacity $T_c$ varies with relay-buffer size $B$ under two network scenarios of $(n = 72,m = 6)$ and $(n = 200,m = 10)$, where $\alpha$ is fixed as $0.5$. We can see from Fig.~\ref{fig:tc_vs_B} that just as discussed in Corollary~\ref{corollary:B_grow}, the throughput capacity of a buffer-limited MANET can be improved by adopting a larger relay-buffer in such network. A careful observation of Fig.~\ref{fig:tc_vs_B} shows that as $B$ increases the capacity $T_c$ first increases quickly and then gradually converges to a constant determined by Corollary~\ref{corollary:tc_infinity}. This observation indicates we can determine a suitable buffer size $B$ according to the requirement on network capacity such that a graceful trade-off between capacity performance and buffer cost can be achieved. 

To further illustrate the impact of buffer constraint on the throughput capacity, we show in Fig.~\ref{fig:tc_vs_n} the relationship between $T_c$ and $n$ under three typical relay-buffer settings, i.e., $B$ is fixed as a constant ($5$ here), $B = n$ and $B \to \infty$.  We can see from Fig.~\ref{fig:tc_vs_n} that in general $T_c$ decreases as $n$ increases, but as $n \to \infty$, $T_c$ vanishes to $0$ when $B$ is fixed, while a non-zero constant throughput capacity can still be achieved when $B = n$ or $B \to \infty$. These behaviors are expected, since the result in Corollary~\ref{corollary:scaling_law} indicates that to achieve a non-vanishing throughput capacity in the concerned MANET, its relay-buffer size $B$ should grow at least linearly with $n$.

\begin{figure}[!t]
    \centering
    \includegraphics[width=3.0in]{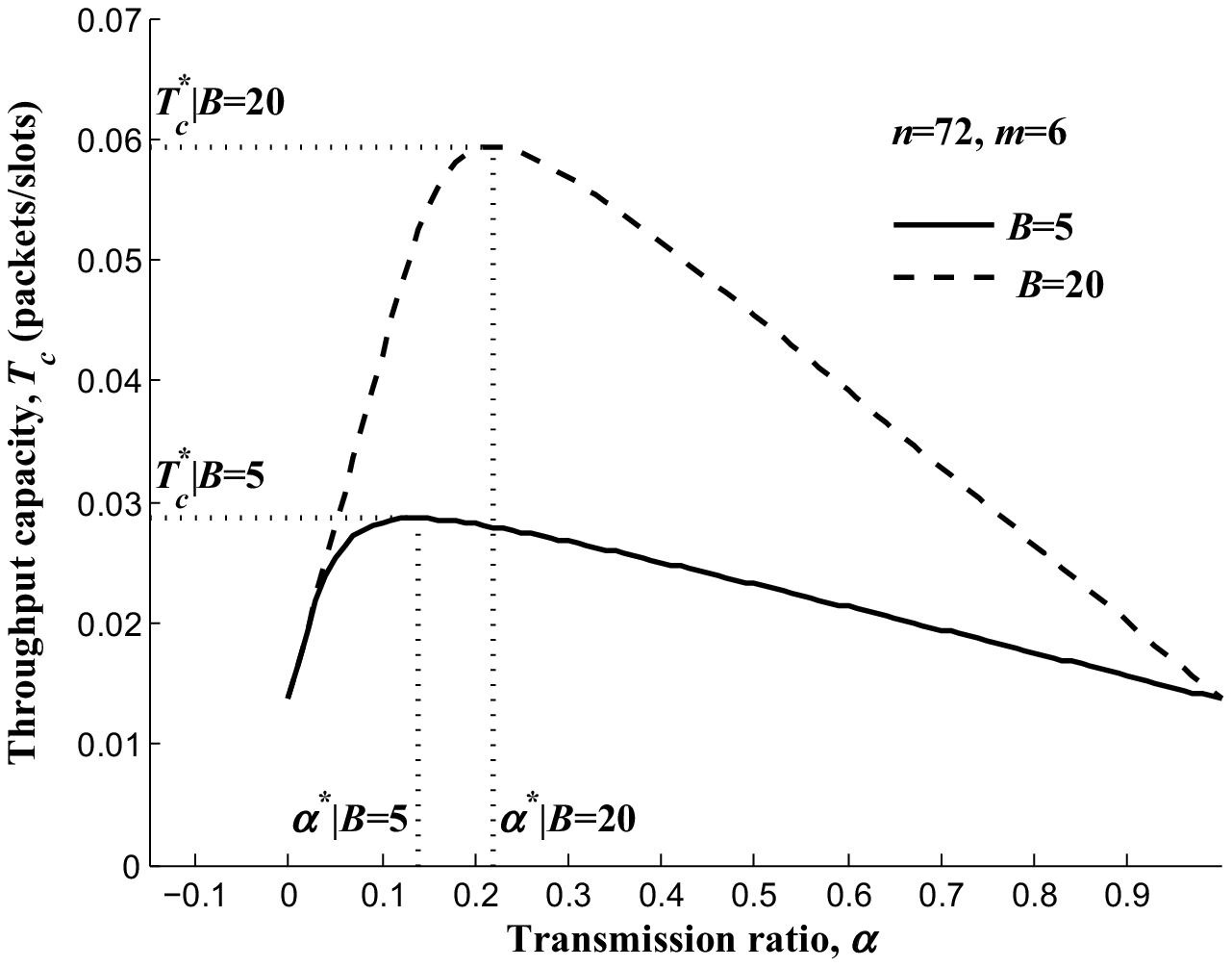}
    \caption{Throughput capacity $T_c$ vs. transmission ratio $\alpha$.}
		\label{fig:tc_vs_alpha}
\end{figure}

\subsection{Optimal Throughput Capacity} \label{subsection:optimal_tc}

\begin{figure}[!t]
    \centering
    {
    \subfloat[Optimal transmission ratio $\alpha^*$ vs. relay-buffer size $B$.]
    {\includegraphics[width=3.0in]{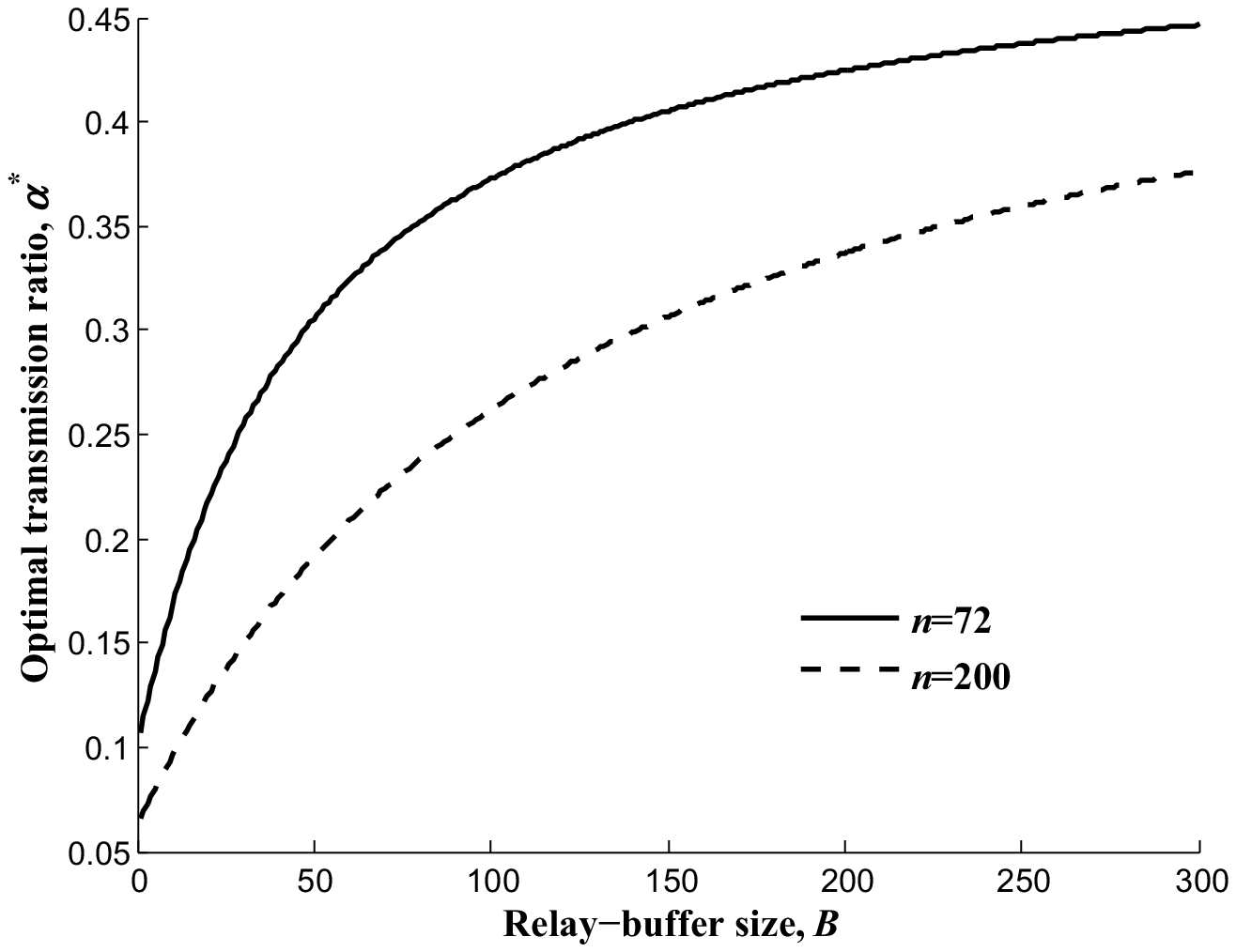} \label{subfig:alphaopt_vs_B}}
   	\hspace{0.15in}
    \subfloat[Optimal transmission ratio $\alpha^*$ vs. number of nodes $n$.]
    {\includegraphics[width=3.0in]{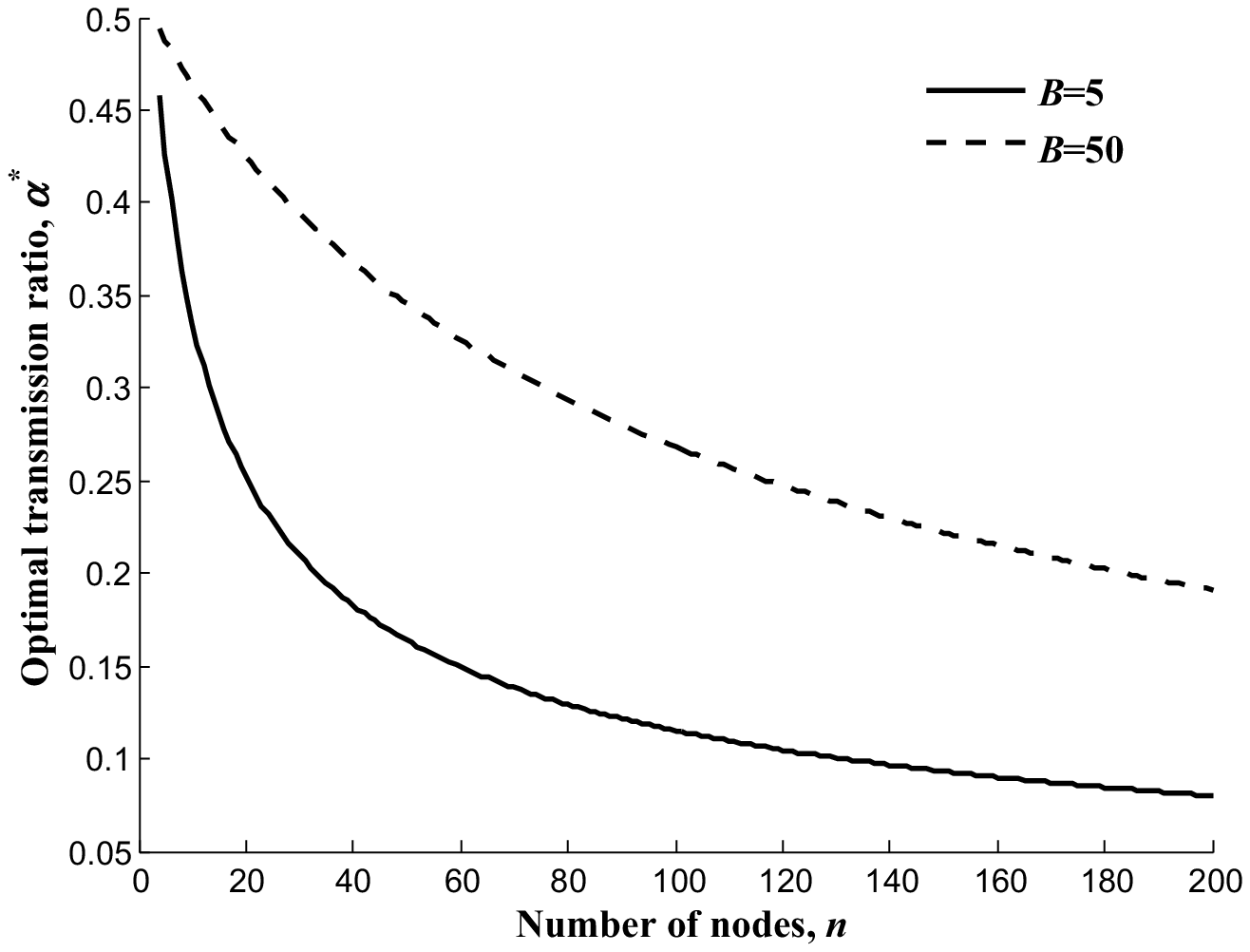} \label{subfig:alphaopt_vs_n}}
    }
    \caption{Optimal transmission ratio $\alpha^*$ vs. $B$ and $n$.}
    \label{fig:optimal_ratio}
\end{figure}

To illustrate the optimal throughput capacity, we show in Fig.~\ref{fig:tc_vs_alpha} the impact of transmission ratio $\alpha$ on throughput capacity $T_c$ under the settings of $n=72$, $m=6$ and $B=\{5,20\}$. We can see from Fig.~\ref{fig:tc_vs_alpha} that under a given setting of $B$, as $\alpha$ increases $T_c$ first increases and then decreases, and just as discussed in Theorem~\ref{theorem:optimization} that there exists an optimal $\alpha^*$ to achieve the optimal throughput capacity $T_c^*$ . This is mainly due to the reason that the effects of $\alpha$ on $T_c$ are two folds. On one hand, a larger $\alpha$ will lead to a higher probability of conducting S-R transmission; on the other hand, a larger $\alpha$ will result in a higher RBP thus a lower opportunity of conducting the S-R transmission. As a summary, in order to improve the throughput performance of a buffer-limited MANET, nodes should cooperate with each other, and they should be neither too selfish nor too selfless.

Based on the results of Theorem~\ref{theorem:optimization}, we illustrate in Fig.~\ref{fig:optimal_ratio} how the optimal transmission ratio $\alpha^*$ is related to $B$ and $n$. We can see that just as proved in Corollary~\ref{corollary:optimal_alpha} that $\alpha^*$ increases as $B$ grows while it decreases as $n$ grows, and the optimal transmission ratio never exceeds $0.5$. These behaviors indicate that in a network with the fixed number of nodes $n$, if we upgrade the capacity of each node by adopting a larger relay-buffer, we should accordingly allocate a higher probability for S-R transmission (i.e., nodes should be more selfish), to achieve the optimal throughput capacity. On the other hand, when the relay-buffer size of each node is fixed, if we increase the scale of the network by accommodating more nodes, we should accordingly increase the probability of R-D transmission (i.e., nodes should be more selfless), to release the relay-buffer space and thus guarantee the optimal throughput capacity there.

\section{Related Work} \label{section:related_work}
A significant amount of work has been devoted to the performance analysis of wireless ad hoc networks, among which some initial studies focused on the buffer-limited network scenarios.

Herdtner and Chong \cite{Herdtner_INFOCOM05} explored the throughput-storage tradeoff in MANETs and showed that the per node capacity under the finite buffer constraint cannot achieve $\Theta(1)$ even though node mobility is utilized. Later, Eun and Wang \cite{Eun_TON08} considered TCP/AQM system and developed a doubly-stochastic analytical model to study the tradeoff among link utilization, packet loss, and buffer size. Subramanian \emph{et al.} \cite{Subramanian_WIPOT09} developed a theoretical framework for throughput analysis in  Delay Tolerant Networks (DTNs). They derived closed-form expressions for the per node throughput capacity of a sparse network which consists of one source-destination pair and several mobile relay nodes with limited buffer size, and further extended their results to the multi-cast network scenarios \cite{Subramanian_ISIT09}. Wang \emph{et al.} \cite{Wang_TPDS13} studied the node buffer occupancy behaviors in static random wireless networks with intermittent connectivity and provided some scaling results of fundamental achievable lower bound for the occupied buffer size. Krifa \emph{et al.} \cite {Krifa_SECON08} focused on the buffer management policies in DTNs. They demonstrated that drop-tail and drop-front policies are sub-optimal and meanwhile proposed an optimal buffer management policy based on global knowledge about the network. Following this line, Elwhishi \emph{et al.} \cite{Elwhishi_TPDS13} developed a new message scheduling framework for Epidemic \cite{Vahdat_00} and Spray-Wait \cite{Spyropoulos_SIGCOMM05} forwarding routings in DTNs with finite buffer to optimize either the message delivery ratio or message delivery delay.

\section{Conclusion} \label{section:conclusion}
In this paper, we first revealed the inherent relationship between the throughput capacity and relay-buffer blocking probability in a buffer-limited MANET with the 2HR-$\alpha$ scheme, and then developed an Embedded Markov Chain-based framework to fully characterize the complicated packet delivery processes of the concerned MANET. Based on this framework, we derived the throughput capacity in closed form and further conduct cases studies under two typical transmission scheduling schemes to illustrate the impacts of some key network parameters on throughput capacity. It is expected the theoretical framework developed in this paper will be also helpful for exploring the throughput capacity of buffer-limited MANETs under other mobility models and other transmission schemes. The results in this paper indicate that in large scale MANETs a non-zero constant throughput capacity can still be guaranteed as long as its relay-buffer size grows at least linearly with network size. Another interesting finding of this paper is that for throughput capacity optimization in such MANETs, the optimal setting of transmission ratio in the 2HR-$\alpha$ scheme there increases with the relay-buffer size but decreases with the network size, and it never exceeds $0.5$.

Notice that the theoretical framework and closed-form results for per node throughput capacity developed in this paper is based on the i.i.d mobility model, so one of our future research directions is to develop theoretical models for other more realistic mobility models, like the inter-meeting time based mobility model.

\appendix

\section{Proof of Lemma~\ref{lemma:transition}} \label{appendix:transition}
Notice that each local queue is a Bernoulli/Bernoulli queue with exogenous packet arrival rate $\lambda$, so the output process of the local queue is also a Bernoulli flow with rate $\lambda$ due to the reversible property of Bernoulli/Bernoulli queue \cite{Daduna_BOOK01}.  

Based on the property of the i.i.d mobility model we know that for a specific node, except itself and its
destination, each of the remaining $n-2$ nodes will equal likely to serve as its relay. Similarly, for each node serving as a relay, except itself and its source, all the remaining $n-2$ nodes will equal likely to forward packets to it. From (\ref{eq:mu_s}) we know that the ratio of S-R transmission to the service rate $\mu_s(\lambda)$ is $\frac{p_{sr}\left(1-P_{b}(\lambda)\right)}{\mu_{s}(\lambda)}$. Hence, the packet arrival rate at a relay queue $\lambda_R$ is determined as

\begin{align}
\lambda_R &= (n-2)\lambda \cdot \frac{p_{sr}\left(1-p_b(\lambda)\right)}{\mu_{s}(\lambda)} \big/ (n-2) \nonumber \\
& = \rho_s(\lambda) p_{sr} \left( 1-p_b(\lambda) \right).
\label{eq:lambda_R}
\end{align}

We denote by $p_{(i,k),\Omega_{i+1}}$ the transition probability from state $(i,k)$ to set $\Omega_{i+1}$, $0 \leq i <B$. Since when a relay queue is full, its packet arrival rate will be $0$. Thus, we have
\begin{align}
& p_{(i,k),\Omega_{i+1}} \cdot \left( 1-p_b(\lambda) \right) + 0 \cdot p_b(\lambda)=\lambda_R,  \nonumber \\
\Rightarrow & p_{(i,k),\Omega_{i+1}}=\rho_s(\lambda) \cdot p_{sr}. \label{eq:lambda_r}
\end{align}

Notice that the transition probability $p_{\Omega_i,\Omega_j}$ of the EMC is the ``set-averaged'' transition probability of the original Markov chain \cite{Hachigian_1963}, then we have

\begin{equation}
p_{\Omega_i,\Omega_j}=\sum_{k=1}^i {p_{(i,k),\Omega_j} \cdot P\left((i,k) \big| \Omega_i \right)}, 
\label{eq:trans}
\end{equation}
where $P\left((i,k) \big| \Omega_i \right)$ is the conditional probability that relay queue is in state $(i,k)$ given that it belongs to the set $\Omega_i$. Substituting (\ref{eq:lambda_r}) into (\ref{eq:trans}) we have
\begin{equation}
p_{\Omega_i,\Omega_{i+1}}=\rho_s(\lambda) \cdot p_{sr}. \nonumber
\end{equation}

A node gets a R-D transmission opportunity with probability $p_{rd}$, and due to the i.i.d mobility model, this opportunity arises for each of the $n-2$ destination nodes with equal probability. Thus, we have 
\begin{equation}
p_{(i,k),\Omega_{i-1}}=k \cdot \frac{p_{rd}}{n-2}.
\label{eq:transition_backward}
\end{equation}
Based on (\ref{eq:trans}), it is clear that in order to obtain $p_{\Omega_i,\Omega_{i-1}}$, we should derive the conditional probability $P \left ( (i,k) \big| \Omega_i \right )$. To address this issue, we utilize the \emph{Occupancy} approach (Chapter 1 in \cite{Stark_BOOK02}). When the relay queue has $i$ packets being buffered, where each packet may be destined for any one of the $n-2$ destination nodes, the number of all possible cases $E_{\Omega_i}$ is given by

\begin{equation}
E_{\Omega_i}=\binom{n-3+i}{i}.
\label{eq:events_all}
\end{equation}
Suppose that these $i$ packets are destined for $k$  distinct destination nodes, the number of cases $E_{(i,k)}$ is given by
\begin{equation}
E_{(i,k)}=\binom{n-2}{k}\cdot \binom{(k-1)+(i-k)}{i-k}.
\label{eq:events_part}
\end{equation}
Due to the i.i.d mobility model, each of these cases occurs with equal probability. According to the \emph{Classical Probability}, $P \left ( (i,k) \big| \Omega_i \right )$ is then determined as 

\begin{equation}
P \left ( (i,k) \big| \Omega_i \right )=\frac{E_{(i,k)}}{E_{\Omega_i}}
=\frac{\binom{n-2}{k}\cdot\binom{i-1}{k-1}}{\binom{n-3+i}{i}}.
\label{eq:conditional_pro}
\end{equation}
It can be easily verified that $\sum \limits_{k\leq i} {P \left ( (i,k) \big| \Omega_i \right )}=1$.
Substituting (\ref{eq:conditional_pro}) into (\ref{eq:trans}) we have

\begin{align}
p_{\Omega_i,\Omega_{i-1}} &=\sum_{k=1}^{i} \left\{ \frac{ \binom{n-2}{k} \cdot \binom{i-1}{k-1}}
{\binom{n-3+i}{i}}\cdot\frac{k p_{rd}}{n-2}\right\} \nonumber \\
&= \frac{p_{rd}}{\binom{n-3+i}{i}} \cdot \sum_{k=0}^{i-1}\left\{ \binom{n-3}{k} \cdot \binom{i-1}{k}\right\} \nonumber \\
&= p_{rd} \cdot \frac{\binom{n-4+i}{i-1}}{\binom{n-3+i}{i}} = \frac{i}{n-3+i} \cdot p_{rd}. \nonumber
\end{align}

\section{Proof of Lemma~\ref{lemma:limit_distribution}} \label{appendix:limit_distribution}
Based on the Lemma~\ref{lemma:transition}, the one-step transition probability matrix $\mathbf{P}$ of the EMC is given by
\begin{equation}       
\mathbf{P}=\left[                 
\begin{array}{cccc}   
p_{\Omega_0,\Omega_0} & p_{\Omega_0,\Omega_1} &   &    \\  
p_{\Omega_1,\Omega_0} & p_{\Omega_1,\Omega_1} &  p_{\Omega_1,\Omega_2}  &    \\  
                      & \ddots & \ddots & \ddots  \\

  &  & p_{\Omega_B,\Omega_{B-1}} & p_{\Omega_B,\Omega_B}
\end{array}
\right].    \label{eq:transition_matrix}             
\end{equation}

From the transition matrix (\ref{eq:transition_matrix}) and the state machine in Fig.~\ref{fig:embedded_markov}, we can see that: 1) the EMC is \emph{irreducible}; 2) each state $\Omega_i$ is \emph{recurrent}, $0 \leq i \leq B$; 3) the period of each state $\Omega_i$ is $1$, so each state is \emph{aperiodic}. Based on these properties, we know that the EMC is ergodic, so its limiting distribution exists and is unique, which is just the stationary distribution of the EMC \cite{Haggstrom_BOOK02}. Thus, we have

\begin{align}
            & \mathbf{\Pi}\cdot\mathbf{P}=\mathbf{\Pi}, \nonumber \\
\Rightarrow & \left\{ 
\begin{aligned}
&\rho_s(\lambda) \cdot p_{sr} \cdot \pi(\Omega_0)=\frac{p_{rd}}{n-2} \cdot \pi(\Omega_1), \\
&\rho_s(\lambda) \cdot p_{sr} \cdot \pi(\Omega_1)=\frac{2\cdot p_{rd}}{n-1} \cdot \pi(\Omega_2), \\
&\cdots, \\
&\rho_s(\lambda) \cdot p_{sr} \cdot \pi(\Omega_{B-1})=\frac{B\cdot p_{rd}}{n-3+B} \cdot \pi(\Omega_B),
\end{aligned}
\right.  \\
\Rightarrow & \pi(\Omega_i) =\mathrm{C}_i \cdot \beta^i \cdot \rho_s(\lambda)^i \cdot \pi(\Omega_0),
\label{eq:balance_equation}
\end{align}

Combining (\ref{eq:balance_equation}) with the normalization equation $\sum_{i=0}^{B}{\pi(\Omega_i)}=1$, the results (\ref{eq:pi_0}) and (\ref{eq:pi_i}) then follow.


\section{Proofs of Corollaries~\ref{corollary:B_grow}, \ref{corollary:tc_alpha_0.5} and \ref{corollary:tc_infinity}} \label{appendix:corollaries}

\textbf{Proof of Corollary~\ref{corollary:B_grow}:} Let $s_k=\frac{\mathrm{C}_k \cdot \beta^k} {\sum_{i=0}^k{\mathrm{C}_i \cdot \beta^i}}$, then 

\begin{align}
\frac{s_{k+1}}{s_k} &= \frac{\mathrm{C}_{k+1} \beta^{k+1} \sum_{i=0}^k{\mathrm{C}_i \cdot \beta^i}}{\mathrm{C}_k \beta^k \sum_{i=0}^{k+1} {\mathrm{C}_i \cdot \beta^i}} \nonumber \\
 & = \frac{\sum_{i=0}^k{ (n-2+k) \mathrm{C}_i \cdot \beta^{i+1}}}{1+k+\sum_{i=0}^k{(k+1)\mathrm{C}_{i+1} \cdot \beta^{i+1}}}, \nonumber
\end{align}  
Since
\begin{equation}
(k+1)\mathrm{C}_{i+1}=(k+1)\frac{n-2+i}{i+1} \cdot \mathrm{C}_i,  \nonumber
\end{equation}
and
\begin{align}
& (k+1)(n-2+i)-(i+1)(n-2+k) \nonumber \\
= &(n-3)(k-i) \geq 0, \nonumber
\end{align}
then
\begin{align}
&(k+1)\mathrm{C}_{i+1} \geq (n-2+k) \mathrm{C}_i, \nonumber \\
&\frac{s_{k+1}}{s_k} < 1, \nonumber
\end{align}
Substituting the result into $(\ref{eq:throughput_capacity})$, the Corollary~\ref{corollary:B_grow} then follows.

\textbf{Proof of Corollary~\ref{corollary:tc_alpha_0.5}:} When $\alpha=0.5$, then $\beta=1$, and (\ref{eq:throughput_capacity}) is simplified as

\begin{equation}
T_c=p_{sd}+p_{sr} \left( 1-\frac{\mathrm{C}_B}{\sum_{i=0}^B {\mathrm{C}_i}}\right).
\label{eq:tc_1}
\end{equation}
Since
\begin{align}
&\sum_{i=0}^B {\mathrm{C}_i} =\frac{1}{(n-3)!} \left[(n-3)\times (n-4) \cdots \times 1 \right. \nonumber \\
& +\left. (n-2)\times \cdots \times 2 + \cdots + (n-3+B)\times \cdots \times (B+1) \right]  \nonumber \\
&=\frac{1}{(n-3)!} \cdot \frac{(n-2+B)\times \cdots \times (B+1)}{n-2} \nonumber \\
&=\binom{n-2+B}{B}, \label{eq:sum_simple}
\end{align}
substituting (\ref{eq:sum_simple}) into (\ref{eq:tc_1}), the Corollary~\ref{corollary:tc_alpha_0.5} then follows. 

\textbf{Proof of Corollary~\ref{corollary:tc_infinity}:} For the case $\alpha=0.5$, since $\lim\limits_{B\to\infty}{\frac{B}{n-2+B}}=1$, substituting it into (\ref{eq:tc_alpha_0.5}) we have 
\begin{equation}
\mathop{T_c}\limits_{\alpha=0.5, B\to \infty}=p_{sd}+p_{sr}. \nonumber
\end{equation}

For the case $\alpha<0.5$, we have $\beta<1$ and
\begin{align}
&\sum_{i=0}^B {\mathrm{C}_i \cdot \beta^i} =\frac{1}{(n-3)!} \nonumber \\
&\times \left[(n-3)\times \cdots \times 1 \times \beta ^0 +(n-2)\times \cdots \times 2 \times \beta ^1 \right. \nonumber \\
&\left. + \cdots + (n-3+B)\times \cdots \times (B+1) \times \beta ^B \right]  \nonumber \\
&=\frac{1}{(n-3)!} \cdot \left(\beta^{n-3}+\beta^{n-2}+\cdots+\beta^{n-3+B} \right)^{(n-3)}  \nonumber \\
&=\frac{1}{(n-3)!} \cdot \left( \sum_{i=0}^{n-3+B} \beta^i \right)^{(n-3)} \nonumber \\
&=\frac{1}{(n-3)!} \left(\frac{1-\beta^{n-2+B}}{1-\beta} \right) ^{(n-3)},
\label{eq:sum_beta}
\end{align}
where $f(\beta)^{(k)}$ denotes the $k$-th order derivative of $f(\beta)$. Since
\begin{equation}
\lim\limits_{B\to\infty}{1-\beta^{n-2+B}}=1, \nonumber
\end{equation}
we have
\begin{align}
\lim\limits_{B\to\infty}{\sum_{i=0}^B {\mathrm{C}_i \cdot \beta^i}} & = \frac{1}{(n-3)!} \left(\frac{1} {1-\beta} \right) ^{(n-3)}\nonumber \\
& =\frac{1} {\left(1-\beta \right)^{n-2}},
\end{align}
and then
\begin{align}
\lim\limits_{B\to\infty}{\mathrm{C}_B\beta^B(1-\beta)^{n-2}} & \leq \lim\limits_{B\to\infty}{(B+n)^n\beta^B} \nonumber \\
& \leq \lim\limits_{B\to\infty} {2^n B^n \beta^B}. 
\end{align}
Since 
\begin{equation}
\lim\limits_{x\to\infty}{x^n \beta^x} =\lim\limits_{x\to\infty}{\frac{x^n}{\frac{1}{\beta}^x}} 
=\lim\limits_{x\to\infty}{\frac{n!}{(-\ln{\beta})^n \cdot \frac{1}{\beta}^x } } =0, \nonumber
\end{equation}
substituting it into (\ref{eq:throughput_capacity}) we have
\begin{equation}
\mathop{T_c}\limits_{\alpha<0.5,B\to \infty}=p_{sd}+p_{sr} \nonumber
\end{equation}

For the case $\alpha>0.5$, we have $\beta>1$ and
\begin{align}
1-\frac{\mathrm{C}_B \cdot \beta^B}{\sum_{i=0}^B {\mathrm{C}_i \cdot \beta^i}} &= \frac{\sum_{i=0}^{B-1} {\mathrm{C}_i \cdot \beta^i}}{1+\beta \sum_{i=0}^{B-1} {\mathrm{C}_{i+1} \cdot \beta^i}} \nonumber \\
&= \frac{1}{ \frac{1}{\sum_{i=0}^{B-1} {\mathrm{C}_i \cdot \beta^i}} + \beta \cdot \frac{\sum_{i=0}^{B-1} {\mathrm{C}_{i+1} \cdot \beta^i}}{\sum_{i=0}^{B-1} {\mathrm{C}_i \cdot \beta^i}} }.  \nonumber
\end{align}
Since
\begin{align}
&\lim \limits_{B \to \infty} {\frac{1}{\sum_{i=0}^{B-1} {\mathrm{C}_i \cdot \beta^i}}} =0, \nonumber \\
&\lim \limits_{B \to \infty} {\frac{\sum_{i=0}^{B-1} {\mathrm{C}_{i+1} \cdot \beta^i}}{\sum_{i=0}^{B-1} {\mathrm{C}_i \cdot \beta^i}}} =1, \nonumber
\end{align}
then
\begin{equation}
\lim \limits_{\beta>1,B\to \infty} {1-\frac{\mathrm{C}_B \cdot \beta^B}{\sum_{i=0}^B {\mathrm{C}_i \cdot \beta^i}}} = \frac{1}{\beta}.
\end{equation}
Substituting it into (\ref{eq:throughput_capacity}) we have
\begin{align}
\mathop{T_c}\limits_{\alpha>0.5,B\to \infty} &=p_{sd}+ p_{sr} \frac{1}{\beta} \nonumber \\
&=p_{sd}+ p_{sr}\frac{\alpha}{1-\alpha}= p_{sd}+ p_{rd}. \nonumber 
\end{align}

\section{Proof of Corollary~\ref{corollary:optimal_alpha}} \label{appendix:optimal_alpha}

Considering $\gamma \in (0,1]$, the first order derivative of $g(\gamma)$ is

\begin{align}
g'(\gamma)=\frac{1}{h(\gamma)^2}\cdot  \{ \underbrace{ h(\gamma)[h(\gamma)+\mathrm{C}_B]-(1+\gamma)\mathrm{C}_B h'(\gamma) }_{(a)}  \}. \nonumber
\end{align}
For $\forall n >3$, when $B=1$, $(a)$ is determined as
\begin{equation}
(a)=\gamma(\gamma+n-2)-(1+\gamma)(n-2)=(\gamma^2-1)-(n-3)\leq 0. \nonumber
\end{equation}
When $B=k$, we assume that 
\begin{equation}
(a)=h_k(h_k+\mathrm{C}_k)-(1+\gamma)\mathrm{C}_k h'_k \leq 0 ,\nonumber
\end{equation}
where $h_k$ and $h'_k$ are the abbreviations of $h(\gamma)$ and $h'(\gamma)$ under $B=k$, respectively. When $B=k+1$, we have
\begin{align}
(a)&= h_{k+1}(h_{k+1}+\mathrm{C}_{k+1})-(1+\gamma)\mathrm{C}_{k+1} h'_{k+1} \nonumber \\
&=\gamma\cdot (h_k+\mathrm{C}_k)\cdot [\gamma(h_k+\mathrm{C}_k)+\mathrm{C}_{k+1}] \nonumber \\
& -(1+\gamma)\cdot \mathrm{C}_{k+1}\cdot [h_k+\gamma h'_k+\mathrm{C}_k] \nonumber \\
&=\underbrace{\gamma^2 h_k(h_k+\mathrm{C}_k)}_{(b_1)} + \gamma^2 \mathrm{C}_k (h_k+\mathrm{C}_k) + \underbrace{\gamma \mathrm{C}_{k+1} h_k}_{(c_1)} \nonumber \\
& + \underbrace{\gamma \mathrm{C}_k \mathrm{C}_{k+1}}_{(d_1)} -\underbrace{(1+\gamma)\mathrm{C}_{k+1}h_k}_{(c_2)}-\underbrace{\gamma(1+\gamma)\mathrm{C}_{k+1}h'_k}_{(b_2)} \nonumber \\
&-\underbrace{(1+\gamma)\mathrm{C}_k\mathrm{C}_{k+1}}_{(d_2)}. \nonumber
\end{align}
Since
\begin{align}
(b_1)-(b_2) & = \gamma [ \gamma h_k(h_k+\mathrm{C}_k)-(1+\gamma)\mathrm{C}_{k+1} h'_k] \nonumber \\
& < \gamma [h_k(h_k+\mathrm{C}_k)-(1+\gamma)\mathrm{C}_k h'_k] \leq 0, \nonumber
\end{align}
combining $(c_1)$,$(c_2)$ and $(d_1)$,$(d_2)$ we have
\begin{align}
(a) &< \gamma^2 \mathrm{C}_k (h_k+\mathrm{C}_k) - \mathrm{C}_{k+1} h_k - \mathrm{C}_k \mathrm{C}_{k+1} \nonumber \\
& = (h_k+\mathrm{C}_k) (\gamma^2 \mathrm{C}_k-\mathrm{C}_{k+1})<0. \nonumber
\end{align}
According to the above mathematical induction, we can conclude that $g'(\gamma)<0$ for $\gamma \in (0,1)$ and $g'(1) \leq 0$. Thus, $g(\gamma)$ monotonically decreases when $\gamma \in (0,1]$, so we know that $\gamma^*>1$ and $\alpha^*=\frac{1}{1+\gamma^*}<0.5$.  

For the limiting case $B \to \infty$, from (\ref{eq:tc_infinity}), (\ref{eq:p_sr}) and (\ref{eq:p_rd}) we can easily see that $\alpha^*|_{B \to \infty} = 0.5$ and $T_c^*|_{B \to \infty} = \frac{p_0+p_1}{2d}$.

\section{Proof of Corollary~\ref{corollary:scaling_law}} \label{appendix:scaling_law}

As $n \to \infty$ we have

\begin{align}
p_0&=1-\left(1-\frac{d}{n}\right)^{\frac{n}{d}\cdot d}-d\left(1-\frac{d}{n}\right)^{\frac{n}{d}\cdot d-1} \nonumber \\
& \to  1- e^{-d}-de^{-d}, \nonumber
\end{align}
\begin{align}
p_1 &=1-\left(1-\frac{d^2}{n^2} \right)^{\frac{n}{2}} =1-e^{\frac{n}{2}\cdot \ln(1-\frac{d^2}{n^2})} \nonumber \\
& \to 1-e^0=0, \nonumber
\end{align}
\begin{align}
1-\frac{\mathrm{C}_B \cdot \beta^B}{\sum_{i=0}^B {\mathrm{C}_i \cdot \beta^i}}&=\frac{\sum_{i=0}^{B-1}{\mathrm{C}_i \cdot \beta^i}}{\sum_{i=0}^B {\mathrm{C}_i \cdot \beta^i}} \nonumber \\
& \to \frac{B}{(n-3+B)\beta}. \nonumber
\end{align}
Substituting above results into (\ref{eq:tc_LTS}), we have that as $n \to \infty$
\begin{equation}
T_c \to  \frac{\alpha (1- e^{-d}-de^{-d})}{d\beta}\frac{B}{(n-3+B)}=\Theta \left(\frac{B}{n} \right). \nonumber
\end{equation}

\end{document}